\documentclass[a4paper,english]{lipics-v2018}
\usepackage{microtype}

\usepackage{amsmath}
\usepackage{amssymb,array,amsthm}

\usepackage{paralist}
\usepackage{tabularx}
\usepackage{changepage}
\usepackage{enumerate}
\usepackage{todonotes}
\usepackage{graphicx}
\usepackage{algorithm}
\usepackage[noend]{algorithmic}
\usepackage{tikz}
\usepackage{bm}
\usepackage{hyperref} % À enlever pour la soumission à arXiv.
\usetikzlibrary{decorations.markings,circuits,decorations.pathmorphing,decorations.pathreplacing,arrows,automata,}
\usepackage{subcaption}
\usepackage{enumitem}

\newcommand{\N}{\mathbb{N}}
\newcommand{\R}{\mathbb{R}}

\newcommand{\Z}{\mathbb{Z}}

\newcommand{\game}{\ensuremath{\mathcal{G}}}

\newcommand{\aVal}{\mathrm{aVal}}
\newcommand{\cVal}{\mathrm{cVal}}
\newcommand{\acVal}{\mathrm{acVal}}

\newcommand{\outcome}{\ensuremath{\mathbf{Out}}}
\newcommand{\history}{\ensuremath{\mathbf{Hist}}}

\newcommand{\prefix}{\ensuremath{\subseteq_{\mathsf{pref}}}}

\newcommand{\last}[1]{\mathrm{last}(#1)}
\newcommand{\first}[1]{\mathrm{first}(#1)}

\theoremstyle{definition}

\newtheorem{proposition}[theorem]{Proposition}
\newtheorem{observation}[theorem]{Observation}

\title{Beyond admissibility: Dominance between chains of strategies}

\titlerunning{Dominance betweens chains of strategies}%optional, please use if title is longer than one line

\author{Nicolas Basset}{Universit\'e Grenoble Alpes\\Grenoble, France}{bassetni@univ-grenoble-alpes.fr}{}{}

\author{Isma\"el Jecker}{Department of Computer Science, Universit\'e Libre de Bruxelles\\Brussels, Belgium}{ismael.jecker@ulb.ac.be}{}{}

\author{Arno Pauly}{Department of Computer Science, Swansea University\\Swansea, UK}{arno.m.pauly@gmail.com}{https://orcid.org/0000-0002-0173-3295}{}

\author{Jean-Fran\c{c}ois Raskin}{Department of Computer Science, Universit\'e Libre de Bruxelles\\Brussels, Belgium}{jraskin@ulb.ac.be}{}{ERC inVEST (279499)}

\author{Marie Van den Bogaard}{Department of Computer Science, Universit\'e Libre de Bruxelles\\Brussels, Belgium}{marie.van.den.bogaard@ulb.ac.be}{}{}

\authorrunning{N.~Basset, I.~Jecker, A.~Pauly, J-F.~Raskin, M.~Van den Bogaard}%mandatory. First: Use abbreviated first/middle names. Second (only in severe cases): Use first author plus 'et. al.'

\Copyright{Nicolas Basset, Isma\"el Jecker, Arno Pauly, Jean-Fran\c{c}ois Raskin and Marie Van den Bogaard}%mandatory, please use full first names. LIPIcs license is "CC-BY";  http://creativecommons.org/licenses/by/3.0/

\subjclass{{\bf Theory of computation $\rightarrow$ Solution concepts in game theory}, Theory of computation $\rightarrow$ Automata extensions}% mandatory: Please choose ACM 2012 classifications from https://www.acm.org/publications/class-2012 or https://dl.acm.org/ccs/ccs_flat.cfm . E.g., cite as "General and reference $\rightarrow$ General literature" or \ccsdesc[100]{General and reference~General literature}.

\keywords{dominated strategies;
admissible strategies;
games played on finite graphs;
reactive synthesis;
reachability games;
safety games;
cofinal;
order theory}%mandatory

\category{}%optional, e.g. invited paper

\relatedversion{}%optional, e.g. full version hosted on arXiv, HAL, or other respository/website

\supplement{}%optional, e.g. related research data, source code, ... hosted on a repository like zenodo, figshare, GitHub, ...

\funding{}%optional, to capture a funding statement, which applies to all authors. Please enter author specific funding statements as fifth argument of the \author macro.

\acknowledgements{We thank Gilles Geeraerts and Guillermo A. P\'erez for several fruitful discussions.}%optional

\EventEditors{}
\EventNoEds{2}
\EventLongTitle{}
\EventShortTitle{}
\EventAcronym{}
\EventYear{}
\EventDate{}
\EventLocation{}
\EventLogo{}
\SeriesVolume{}
\ArticleNo{}
\nolinenumbers %uncomment to disable line numbering
\hideLIPIcs  %uncomment to remove references to LIPIcs series (logo, DOI, ...), e.g. when preparing a pre-final version to be uploaded to arXiv or another public repository
\begin{document}

\maketitle

\begin{abstract}
Admissible strategies, i.e.~those that are not dominated by any other strategy, are a typical rationality notion in game theory. In many classes of games this is justified by results showing that any strategy is admissible or dominated by an admissible strategy. However, in games played on finite graphs with quantitative objectives (as used for reactive synthesis), this is not the case.

We consider increasing chains of strategies instead to recover a satisfactory rationality notion based on dominance in such games. We start with some order-theoretic considerations establishing sufficient criteria for this to work. We then turn our attention to generalised safety/reachability games as a particular application. We propose the notion of maximal uniform chain as the desired dominance-based rationality concept in these games. Decidability of some fundamental questions about uniform chains is established.
 \end{abstract}

\section{Introduction}
The canonical model to formalize the reactive synthesis problem are two-player win/lose perfect information games played on finite (directed) graphs~\cite{PnRo89,AbadiLW89}. In recent years, %researchers have found this model too limited, and have studied
more general objectives and multiplayer games have been studied (see e.g.~\cite{kupferman} or \cite{DBLP:conf/lata/BrenguierCHPRRS16} and additional references therein). When moving beyond two-player win/lose games, the traditional solution concept of a \emph{winning strategy} needs to be updated by another notion. The game-theoretic literature offers a variety of concepts of \emph{rationality} to be considered as candidates.

The notion we focus on here is \emph{admissibility}: roughly speaking, judging strategies according to this criterion allows to deem rational only strategies that are not \emph{worse} than any other strategy (ie, that are not \emph{dominated}). In this sense, admissible strategies represent maximal elements in the whole set of strategies available to a player. One attractive feature of admissibility, or more generally, dominance based rationality notions is that they work on the level of an individual agent. Unlike e.g.~to justify Nash equilibria, no common rationality, shared knowledge or any other assumptions on the other players are needed to explain why a specific agent would avoid dominated strategies.

The study of admissibility in the context of games played on graphs
was initiated by Berwanger in \cite{Berwanger07} and subsequently became an active research topic (e.g.~\cite{BRS14,BrenguierPRS16,BGRS17,pauly-raskin2,DBLP:journals/acta/BrenguierRS17}, see related work below).
In \cite{Berwanger07}, Berwanger established in the context of perfect-information games with boolean objectives that admissibility is the \emph{good} criterion for rationality: every strategy is either admissible or dominated by an admissible strategy.

Unfortunately, this fundamental property does not hold when one considers quantitative objectives.
Indeed, as soon as there are three different possible payoffs, one can find instances of games where a strategy is neither dominated by an admissible strategy, nor admissible itself (see Example~\ref{ex:best_animal}).
This third payoff actually allows for the existence of infinite domination sequences of strategies, where each element of the sequence dominates its predecessor and is dominated by its successor in the chain.
Consequently, no strategy in such a chain is admissible.
However, it can be the case that no admissible strategy dominates the elements of the chain.
In the absence of a \emph{maximal element} above these strategies, one may ask why they should be discarded in the quest of a rational choice.
They may indeed represent a type of behaviour that is rational but not captured by the admissibility criterion.

\subparagraph*{Our contributions.}
To formalize this behaviour, we study increasing chains of strategies (Definition \ref{def:dominanceforchains}). A chain is weakly dominated by some other chain, if every strategy in the first is below some strategy in the second. The question then arises whether every chain is below a maximal chain. Based on purely order-theoretic argument, a sufficient criterion is given in Theorem \ref{theo:maximalelements}. However, Corollary \ref{cor:example_uncountable_chains} shows that our sufficient criterion does not apply to all games of interests. We can avoid the issue by restricting to some countable class of strategies, e.g.~just the regular, computable or hyperarithmetic ones (Corollary \ref{corr:countablymanystrategiesgood}).

We test the abstract notion in the concrete setting of generalised safety/reachability games (Definition \ref{def:gen_safety_reach_game}). Based on the observation that the crucial behaviour captured by chains of strategies, but not by single strategies is \emph{Repeat this action a large but finite number of times}, we introduce the notion of a parameterized automaton (Definition \ref{def:parameterized_automaton}), which essentially has just this ability over the standard finite automata. We then show that any finite memory strategy is below a maximal chain or strategy realized by a parameterized automaton (Theorem \ref{theo:dominatedbyuniformchain}).

Finally, we consider some algorithmic properties of chains and parameterized automata in generalised safety/reachability games. It is decidable in $\textsc{PTime}$ whether a parameterized automaton realizes a chain of strategies (Theorem~\ref{theo:assumption_1}). It is also decidable in $\textsc{PTime}$ whether the chain realized by one parameterized automaton dominates the chain realized by another (Theorem~\ref{theo:uniform_dominance}).

\subparagraph*{Related work.}
As mentioned above, the study of dominance and admissibility for games played on graphs %in extensive form
was initiated by Berwanger in \cite{Berwanger07}. Faella analyzed several criteria for how a player should play a win/lose game on a finite graph that she cannot win, eventually settling on the notion of admissible strategy \cite{Faella09}.

Admissibility in quantitative perfect-information sequential games
played on graphs
was studied in~\cite{BrenguierPRS16}. Concurrent games were considered in \cite{BGRS17}. In \cite{pauly-raskin2}, games with imperfect information, but boolean objectives were explored. The study of decision problems related to admissibility (as we do in Subsection \ref{subsec:algorithms}) was advanced in \cite{BRS14}. The complexity of decision problems related to dominance in normal form games has received attention, see \cite{pauly-dominance} for an overview. For the role of admissibility for synthesis, we refer to \cite{DBLP:journals/acta/BrenguierRS17}.

Our Subsection \ref{subsec:orderingchains} involves an investigation of cofinal chains in certain quasi-ordered sets. A similar theme (but with a different focus) is present in \cite{shangzhi}.

%\paragraph*{Structure of the paper}

\section{Background}
\label{sec:background}
\subsection{Games on finite graphs}

%We denote by $\R$ the set of real numbers and by $\N$ the set of natural numbers.

%\medskip

\noindent A \emph{turn-based multiplayer game}  $\game$ on a finite graph $G$  is a tuple $\game = \langle P, G , (p_i)_{i \in P} \rangle$ where:
\begin{itemize}
\item $P$ is the non-empty finite set of players of the game,
\item $G= \langle V , E \rangle $ where the finite set $V$ of vertices of $G$ is equipped with a $\vert P \vert$-partition $\uplus_{i \in P} V_i$, and $E \subseteq V \times V$ is the edge relation of $G$,
\item for each player $i$ in $P$, $p_i$ is a \emph{payoff function} that associates to every infinite path in $G$ a  payoff in $\R$.
\end{itemize}

\subparagraph*{Outcomes and histories.}
An \emph{outcome} $\rho$ of $\game$ is an infinite path in $G$, that is, an infinite sequence of vertices $\rho= (\rho_k)_{k\in \N} \in V^{\omega}$, where for all $k  \in \N$, $(\rho_k, \rho_{k+1}) \in E$.
The set of all possible outcomes in $\game$ is denoted $\outcome(\game)$.
A finite prefix of an outcome is called a \emph{history}.
The set of all histories in $\game$ is denoted $\history(\game)$.
For an outcome $\rho= (\rho_k)_{k\in \N}$ and an integer $\ell$, we denote by $\rho_{\leq \ell}$ the history $(\rho_k)_{0 \leq k \leq \ell}$.
The \emph{length} of the history $\rho_{\leq \ell}$, denoted $\vert \rho_{\leq \ell} \vert$ is $\ell +1$.
Given an outcome or a history $\rho$ and a history $h$, we write $h \prefix \rho$ if $h$ is a prefix of $\rho$, and we denote by $h^{-1}.\rho$ the unique outcome (or history) such that $\rho = h. (h^{-1}.\rho)$.
Given an outcome $\rho$ or a history $h$ and $k \in \N$ (respectively $k < |h|$), we denote by $\rho_k$ (respectively $h_k$) the $k+1$-th vertex of $\rho$ (respectively of $h$).
For a history $h$, we define the last vertex of $h$ to be $\last h := h_{\vert h \vert -1}$ and its first vertex $\first h := h_0$.
%Given two outcomes or histories $\rho$ and $\rho'$, the \emph{longest common prefix} of $\rho$ and $\rho'$ is denoted $\texttt{lcp}(\rho, \rho')$.
For a vertex $v \in V$, its set of \emph{successors} is $E_v= \lbrace v' \in V ~\vert~ (v,v') \in E \rbrace$.

\medskip
 \subparagraph*{Strategy profiles and payoffs.}
A \emph{strategy} of a player~$i$ is a function $\sigma_i$ that associates to each history $h$ such that $\last h \in V_i$, a successor state $v  \in E_{\last h}$.
A tuple of strategies $(\sigma_i)_{i \in P'}$ where $P' \subseteq P$, one for each player in $P'$ is called a \emph{profile} of strategies.
Usually, we focus on a particular player~$i$, thus, given a profile $(\sigma_i)_{i \in P}$, we write $\sigma_{-i}$ to designate the collection of strategies of players in $P\setminus \lbrace i \rbrace$, and the complete profile is written  $(\sigma_i, \sigma_{-i})$.
The set of all strategies of player $i$ is denoted $\Sigma_i(\game)$, while $\Sigma(\game)= \prod_{i \in P} \Sigma_i(\game)$ is the set of all profiles of strategies in the game $\game$ and $\Sigma_{-i}(\game)$ is the set of all profiles of all players except Player~$i$.
As we consider games with perfect information and deterministic transitions, any complete profile $\sigma_P= (\sigma_i)_{i \in P}$ yields,  from any history $h$,  a unique outcome, denoted $\outcome_h( \game, \sigma_P)$.
Formally, $\outcome_h( \game, \sigma_P)$ is the outcome $\rho$ such that $\rho_{\leq |h|-1} = h$ and  for all $k \geq |h|-1$, for all $i \in P$, its holds that $\rho_{k+1} = \sigma_i(\rho_{\leq k})$ if $\rho_k \in V_i$.
The set of outcomes (resp. histories) \emph{compatible} with a strategy $\sigma$ of player~$i$ after a history $h$ is $\outcome_h(\game, \sigma_i) = \lbrace \rho \in  ~\vert~ \exists \sigma_{-i} \in \Sigma_{-i}(\game) \text{ such that }  \rho = \outcome_h( \game,  (\sigma_i,  \sigma_{-i})) \rbrace$ (resp. $\history_h(\sigma) = \lbrace h \in \history(\game) ~\vert ~  \exists \rho \in \outcome_h(\game, \sigma_i), n \in \N \text{ such that } h=\rho_{\leq n} \rbrace$).
Each outcome $\rho$ yields a \emph{payoff} $p_i(\rho)$ for each Player~$i$.
%When the player considered is clear, we use the notation $\langle \rho \rangle$.
We denote with $p_i(h,\sigma, \tau)$ the payoff of a profile of strategies $(\sigma, \tau)$ after a history $h$.
\iffalse
Given two strategies $\sigma, \sigma'$ of Player~$i$, and a history $h$, $\sigma[h \leftarrow \sigma']$ denotes the strategy that follows $\sigma$ and \emph{shifts} to $\sigma'$ at history $h$.
Formally, for all histories $h' \in \history(\game)$, the strategy $\sigma[h \leftarrow \sigma']$ is such that:
\[ \sigma[h \leftarrow \sigma'](h') =
\begin{cases}
\sigma'(h') \text{~if~} h \prefix h', \\
\sigma(h') \text{~otherwise~}.

\end{cases}
\]
\fi

Usually, we consider games instances such that players start to play at a fixed vertex. Thus, we call an \emph{initialized game} a pair $(\game, v_0)$ of a game $\game$ and a vertex $v_0 \in V$.
When the initial vertex $v_0$ is clear from context, we speak directly from $\game$, $\outcome(\game, \sigma_P)$  and  $p_i(\sigma_P)$ instead of $(\game, v_0)$, $\outcome_{v_0}(\game, \sigma_P)$ and  $p_i(v_0,\sigma_P)$.

\noindent \subparagraph*{Dominance relation.}\label{def:strat_dominance}
In order to compare different strategies of a player~$i$ in terms of payoffs, we rely on the notion of dominance between strategies:
A strategy $\sigma \in \Sigma_i$ is \emph{weakly dominated} by a strategy $\sigma' \in \Sigma_i$
\emph{ at a history $h$} compatible with $\sigma$ and $\sigma'$, denoted $ \sigma \preceq_h \sigma'$, if for every $\tau \in \Sigma_{-i}$, we have $p_i(h,\sigma, \tau) \leq p_i(h,\sigma', \tau)$.
 We say that $\sigma$ is weakly dominated by $\sigma'$, denoted $\sigma \preceq \sigma'$ if $\sigma \preceq_{v_0} \sigma'$, where $v_0$ is the initial state of $\game$.
 A strategy $\sigma \in \Sigma_i$ is \emph{dominated} by a strategy $\sigma' \in \Sigma_i$,
\emph{ at a history $h$} compatible with $\sigma$ and $\sigma'$,
 denoted $ \sigma \prec_h \sigma'$, if $ \sigma \preceq_h \sigma'$  and there exists $\tau \in \Sigma_{-i}$, such that $p_i(h,\sigma, \tau) < p_i(h,\sigma', \tau)$.
 We say that $\sigma$ is dominated by $\sigma'$, denoted $\sigma \prec \sigma'$ if $\sigma \prec_{v_0} \sigma'$, where $v_0$ is the initial state of $\game$.
 Strategies that are not dominated by any other strategies are called \emph{admissible}:
 A strategy $\sigma \in \Sigma_i$ is \emph{admissible} (respectively \emph{from $h$}) if $\sigma \not\prec \sigma'$ (resp. $\sigma \not\prec_h \sigma'$)  for every $\sigma' \in \Sigma_i$.

\noindent \subparagraph*{Antagonistic and Cooperative Values}

To study the rationality of different behaviours in a game $\game$, it is useful to be able to know, for a player~$i$, a fixed strategy $\sigma \in \Sigma_i$ and any history $h$, the worst possible payoff Player~$i$ can obtain with $\sigma$ from $h$ (i.e., the payoff he will obtain assuming the other players play \emph{antagonistically}), as well as the best possible payoff Player~$i$ can hope for with $\sigma$ from $h$ (i.e., the payoff he will obtain assuming the other players play \emph{cooperatively}).
The first value is called the \emph{antagonistic value of the strategy} $\sigma$ of Player~$i$ at history $h$ in $\game$ and the second value is called the \emph{cooperative value of the strategy} $\sigma$ of Player~$i$ at history $h$ in $\game$.
They are formally defined as $ \aVal_i(\game, h, \sigma) := \inf_{\tau \in \Sigma_{-i}} p_i(\outcome_h(\sigma,\tau))$ and $\cVal_i(\game, h, \sigma) := \sup_{\tau \in \Sigma_{-i}} p_i(\outcome_h(\sigma,\tau))$.

Prior to any choice of strategy of Player~$i$, we can define, for any history $h$, the \emph{antagonistic value of $h$} for Player~$i$ as $\aVal_i(\game,h) := \sup_{\sigma \in \Sigma_i} \aVal_i(\game,h,\sigma)$ and the \emph{cooperative value of $h$} for Player~$i$ as $\cVal_i(\game,h) := \sup_{\sigma \in \Sigma_i} \cVal_i(\game,h,\sigma)$.
Furthermore, one can ask, from a history $h$, what is the maximal payoff one can obtain \emph{while} ensuring the antagonistic value of $h$.
Thus, we define the \emph{antagonistic-cooperative value of $h$} for Player~$i$ as $acVal_i(\game, h) := \sup \lbrace \cVal_i(\game,h,\sigma) ~|~ \sigma\in\Sigma_i \text{~and~} \aVal_i(\game, h ,\sigma) \geq \aVal_i(\game,h) \rbrace$.
From now on, we will omit to precise $\game$ when it is clear from the context.
Given a history $h$, we say that a strategy $\sigma$ of Player~$i$ is \emph{worst-case optimal} if $\aVal_i(h, \sigma) = \aVal_i(h)$, that it is \emph{cooperative optimal} if $\cVal_i(h, \sigma) = \cVal_i(h)$ and \emph{worst-case cooperative optimal} if $\aVal_i(h, \sigma) = \aVal_i(h)$ and $\cVal_i(h, \sigma) = acVal_i(h)$.

%\medskip

An initialized game $(\game, v_0)$ is \emph{well-formed for Player~$i$} if, for every history $h \in \history_{v_0}(\game)$, there exists a strategy $\sigma \in \Sigma_i$ such that $\aVal_i(h,\sigma) = \aVal(h)$, \emph{and} a strategy $\sigma' \in \Sigma_i$ such that $\cVal_i(h,\sigma') = \cVal(h)$.
In other words, at every history $h$, Player~$i$ has a strategy that ensures the payoff $\aVal_i(h)$, and a strategy that allows the other players to cooperate to yield a payoff of $cVa_i(h)$.

%\medskip

In the following, we will always focus on the point of view of one player $i$, thus we will sometimes refer to him as the \emph{protagonist} and assume it is the first player, while the other players $-i$ can be seen as a coalition and abstracted to a single player, that we will call the \emph{antagonist}.
Furthermore, we will omit the subscript $i$ to refer to the protagonist when we use the notations $\aVal_i, \cVal_i, acVal_i, p_i$, etc..

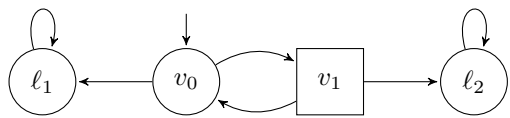
\begin{figure}[t]
\centering
    \begin{tikzpicture}[->,>=stealth',shorten >=1pt, initial text={}]
      \node [initial above ,state] (q1)                      {$v_0$};
      \node[state,rectangle]          (q2) [right=of q1]         {$v_1$};
      \node[state]          (q3) [right=of q2]         {$\ell_2$};
     \node[state]          (q4) [left=of q1]         {$\ell_1$};
      \path (q1) edge [bend left] node {$$} (q2);
      \path (q2) edge [bend left] node {$$} (q1);
      \path (q2) edge node {$$} (q3);
       %\path (q3) edge node {$(-,Yes)$} (q4);
       %\path (q3) edge [bend left] node {$(-,No)$} (q1);
      \path (q1) edge node {$$} (q4);
       \path (q3) edge [loop above] node {$$} (q3);
            \path (q4) edge [loop above] node {$$} (q4);
    \end{tikzpicture}
    \caption{The \emph{Help-me?}-game
    }
    \label{fig:best_animal}
\end{figure}

\begin{example}\label{ex:best_animal}
Consider the game depicted in Figure~\ref{fig:best_animal}.
The protagonist owns the circle vertices.
The payoffs are defined as follows for the protagonist :
\[p(\rho)=
\begin{cases}
0 \text{~if~} \rho=(v_0v_1)^{\omega}, \\
1 \text{~if~} \rho=(v_0v_1)^n v_0 \ell_1^{\omega} \text{~where~} n\in \N, \\
2 \text{~if~} \rho=(v_0v_1)^n \ell_2^{\omega} \text{~where~} n\in \N. \\
\end{cases}
\]
Let us first look at the possible behaviours of the protagonist in this game, when he makes no assumption on the payoff function of the antagonist.
He can choose to be ``optimistic'' and opt to try (at least for some time, or forever) to go to $v_1$ in the hope that the antagonist will cooperate to bring him to $\ell_2$, or settle from the start and go directly to $\ell_1$, not counting on any help from the antagonist.
We denote by $s_k$ the strategy that prescribes to choose $v_1$ as the successor vertex at the first $k$ visits of $v_0$, and $\ell_1$ at the $k+1$-th visit, while $s_{\omega}$ denotes the strategy that prescribes $v_1$ at every visit of $v_0$.
Note that at history $q_0$, the strategy $s_\omega$ is cooperative optimal but not worst-case optimal (as the protagonist takes the risk to get a payoff of $0$ by staying forever in the loop $q_0q_1$), while the strategy $s_0$ that goes directly to $\ell_1$ is worst-case optimal but not cooperative optimal.
On the other hand, strategies $s_k$ for $k>0$ are worst-case cooperative optimal at $q_0$: they allow the antagonist to help reaching $\ell_2$ but also ensure the payoff $1$ by not letting the protagonist loop indefinitely in $q_0q_1$.
Fix $k \in \N$.
Then, $s_k \prec s_{k+1}$:
Indeed, for all $\tau\in \Sigma_{-i}$, if $p( s_k, \tau) = 2$, then there exists $j \leq k$ such that $\tau((v_0v_1)^j)=\ell_2$.
As $s_k$ and $s_{k+1}$ agree up to $(v_0v_1)^k q_0$, we have that $\outcome(s_{k+1},\tau)=(v_0v_1)^j \ell_2^{\omega} = \outcome(s_k,\tau)$, thus $p( s_{k+1}, \tau) = 2$ as well.
Furthermore, consider a strategy $\tau$ such that $\tau((v_0v_1)^j)=v_0$ for all $j \leq k$ and $\tau((v_0v_1)^{k+1})=\ell_2$.
Then $p( s_k, \tau )= 1$ while $p( s_{k+1}, \tau) = 2$.
Finally, consider the strategy $\tau$ such that $\tau((v_0v_1)^k)=v_0$ for all $k \in \N$.
Then $p( s_k, \tau )= 1 = p( s_{k+1}, \tau)$.
Hence, $s_k \prec s_{k+1}$.
In addition, we observe that $s_{\omega}$ is admissible:
for any strategy $s_k$, the strategy $\tau$ of the antagonist that moves to $\ell_2$ at the $k+1$-th visit of $v_1$ yields a payoff of $1$ against strategy $s_k$ but $2$ against strategy $s_\omega$.
Thus, $s_\omega \not \preceq s_k$ for any $k \in \N$.
\end{example}

\subparagraph*{Quantitative vs Boolean setting.}

Remark that in the boolean variant of the \emph{Help-me?} game considered in Example~\ref{ex:best_animal}, where the payoff associated with the vertex $\ell_1$ is $0$ and the payoff associated with the vertex $\ell_2$ is $1$, every strategy $s_k$ for $k \in \N$ is in fact dominated by $s_{\omega}$, as $s_k$ and $s_\omega$ both yield payoff $0$ against $\tau$ such that $\tau((v_0v_1)^k)=v_0$ for all $k \in \N$.
In fact, Berwanger in~\cite{Berwanger07}, showed that boolean games with $\omega$-regular objectives enjoy the following fundamental property: every strategy is either admissible, or dominated by an admissible strategy.
The existence of an admissible strategy in any such game follows as an immediate corollary.
%\medskip

Let us now illustrate how admissibility fails to capture fully the notion of rational behaviour in the quantitative case.
Firstly, recall that the existence of admissible strategies is not guaranteed in this setting (see for instance the examples given in~\cite{BrenguierPRS16}).
In~\cite{BrenguierPRS16}, the authors identified a class of games for which the existence of admissible strategies (for Player~$i$) is guaranteed: well-formed games (for Player~$i$).
However, even in such games, the desirable fundamental property that holds for boolean games is not assured to hold anymore.
In fact, this is already true for quantitative well-formed games with only three different payoffs and really simple payoff functions.
Indeed, consider again the \emph{Help-me?} game in Figure~\ref{fig:best_animal}.
Remark that it is a well-formed game for the protagonist.
We already showed that any strategy $s_k$ is dominated by the strategy $s_{k+1}$.
Thus, none of them is admissible.
The only admissible strategy is $s_\omega$.
It is easy to see that $s_k \not \preceq s_\omega$ for any $k \in \N$:
Let $\tau \in \Sigma_{-i}$ be such that $\tau((v_0v_1)^k)=v_0$ for all $k \in \N$.
Then $p( s_k, \tau ) = 1 > 0 = p( s_{\omega}, \tau )$.
To sum up, we see that there exists an infinite sequence $(s_k)_{k \in \N}$ of strategies such that none of its elements is dominated by the only admissible strategy $s_{\omega}$.
However, the sequence $(s_k)_{k \in \N}$ is totally ordered by the dominance relation.
Based on these observations, we take the approach to not only consider single strategies, but also such ordered sequences of strategies, that can represent a type of rational behaviour not captured by the admissibility concept.
\subsection{Order theory}
In this paragraph we recall the standard results from order theory that we need (see e.g.~\cite{markowsky}). %a more thorough introduction to order theory.

A \emph{linear order} is a total, transitive and antisymmetric relation. A linearly ordered set $(R,\prec)$ is a \emph{well-order}, if every subset of $R$ has a minimal element w.r.t.~$\prec$. The ordinals are the canonical examples of well-orders, in as far as any well-order is order-isomorphic to an ordinal. The ordinals themselves are well-ordered by the relation $<$ where $\alpha \leq \beta$ iff $\alpha$ order-embeds into $\beta$. The first infinite ordinal is denoted by $\omega$, and the first uncountable ordinal by $\omega_1$.

A \emph{quasi order} is a transitive and reflexive relation. Let $(X,\preceq)$ be a quasi-ordered set. A \emph{chain} in $(X,\preceq)$ is a subset of $X$ that is totally ordered by $\preceq$. An \emph{increasing chain} is an ordinal-indexed family $(x_\beta)_{\beta < \alpha}$ of elements of $X$ such that $\beta < \gamma < \alpha \Rightarrow x_\beta \prec x_\gamma$. If we only have that $\beta < \gamma$ implies $x_\beta \preceq x_\gamma$, we speak of a \emph{weakly increasing chain}. We are mostly interested in (weakly) increasing chains in this paper, and will thus occasionally suppress the words \emph{weakly increasing} and only speak about \emph{chains}.

A subset $Y$ of a quasi-ordered set $(X, \preceq)$ is called \emph{cofinal}, if for every $x \in X$ there is a $y \in Y$ with $x \preceq y$. A consequence of the axiom of choice is that every chain contains a cofinal increasing chain, which is one reason for our focus on increasing chains. It is obvious that having multiple maximal elements prevents the existence of a cofinal chain, but even a lattice can fail to admit a cofinal chain. An example we will go back to is $\omega_1 \times \omega$ (cf.~\cite{markowsky}).

If $(X,\preceq)$ admits a cofinal chain, then its \emph{cofinality} (denoted by $\mathrm{cof}(X,\preceq)$) is the least ordinal $\alpha$ indexing a cofinal increasing chain in $(X,\preceq)$. The possible values of the cofinality are $1$ or infinite regular cardinals (it is common to identify a cardinal and the least ordinal of that cardinality). In particular, a countable chain can only have cofinality $1$ or $\omega$. The first uncountable cardinal $\aleph_1$ is regular, and $\mathrm{cof}(\omega_1) = \omega_1$.

We will need the probably most-famous result from order theory:

\begin{lemma}[Zorn's Lemma]
\label{lemma:zorn}
If every chain in $(X,\preceq)$ has an upper bound, then every element of $X$ is below a maximal element.
\end{lemma} 

\section{Increasing chains of strategies}
\label{sec:chains}
\subsection{Ordering chains}
\label{subsec:orderingchains}
In this subsection, we study the quasi-order of increasing chains in a given quasiorder $(X,\preceq)$. We denote by $\mathrm{IC}(X,\preceq)$ the set of increasing chains in $(X,\preceq)$. Our intended application will be that $(X,\preceq)$ is the set of strategies for the protagonist in a game ordered by the dominance relation. However, in this subsection we are not exploiting any properties specific to the game-setting. Instead, our approach is purely order-theoretic.

\begin{definition}
\label{def:dominanceforchains}
We introduce an order $\sqsubseteq$ on $\mathrm{IC}(X,\preceq)$ by defining:
%\[(x_\beta)_{\beta < \alpha} \sqsubseteq (y_\gamma)_{\gamma < \delta} \quad :\Leftrightarrow \quad \forall \beta < \alpha \ \exists \gamma < \delta \ \ x_\beta \preceq y_\gamma \]
\[(x_\beta)_{\beta < \alpha} \sqsubseteq (y_\gamma)_{\gamma < \delta} \quad \text{ if } \quad \forall \beta < \alpha \ \exists \gamma < \delta \ \ x_\beta \preceq y_\gamma \]

Note that $\sqsubseteq$ is a partial order. Let $\doteq$ denote the corresponding equivalence relation. We will occasionally write short $\mathrm{IC}$ for $(\mathrm{IC}(X,\preceq),\sqsubseteq)$.
\end{definition}

Inspired by our application to dominance between strategies in games, we will refer to both $\preceq$ and $\sqsubseteq$ as the \emph{dominance} relation, and might express e.g.~$(x_\beta)_{\beta < \alpha} \sqsubseteq (y_\gamma)_{\gamma < \delta}$ as $(x_\beta)_{\beta < \alpha}$ \emph{is dominated by} $(y_\gamma)_{\gamma < \delta}$, or $(y_\gamma)_{\gamma < \delta}$ \emph{dominates} $(x_\beta)_{\beta < \alpha}$. There is no risk to confuse whether $\preceq$ or $\sqsubseteq$ is meant, since $x \preceq y$ iff $(x)_{\beta < 1} \sqsubseteq (y)_{\gamma < 1}$. Continuing the identification of $x \in X$ and $(x)_{\beta < 1} \in \mathrm{IC}$, we will later also speak about a single strategy dominating a chain or vice versa.

The central notion we are interested in will be that of a maximal chain:
\begin{definition}
$A \in \mathrm{IC}$ is called \emph{maximal}, if $A \sqsubseteq B$ for $B \in \mathrm{IC}$ implies $B \sqsubseteq A$.
\end{definition}

We desire situations where every chain in $\mathrm{IC}$ is either maximal or below a maximal chain. Noting that this goal is precisely the conclusion of Zorn's Lemma (Lemma \ref{lemma:zorn}), we are led to study chains of chains; for if every chain of chains is bounded, Zorn's Lemma applies. Since $(\mathrm{IC}, \sqsubseteq)$ is a quasiorder just as $(X,\preceq)$ is, notions such as cofinality apply to chains of chains just as they apply to chains. We will gather a number of lemmas we need to clarify when chains of chains are bounded.

In a slight abuse of notation, we write $(x_\beta)_{\beta < \alpha} \subseteq (y_\gamma)_{\gamma < \delta}$ iff $\{x_\beta \mid \beta < \alpha\} \subseteq \{y_\gamma \mid \gamma < \delta\}$.
 Clearly, $(x_\beta)_{\beta < \alpha} \subseteq (y_\gamma)_{\gamma < \delta}$ implies $(x_\beta)_{\beta < \alpha} \sqsubseteq (y_\gamma)_{\gamma < \delta}$. We can now express cofinality by noting that $(x_\beta)_{\beta < \alpha}$ is cofinal in $(y_\gamma)_{\gamma < \delta}$ iff $(x_\beta)_{\beta < \alpha} \subseteq (y_\gamma)_{\gamma < \delta}$ and $(y_\gamma)_{\gamma < \delta} \sqsubseteq (x_\beta)_{\beta < \alpha}$.
 We recall that the cofinality of $(y_\gamma)_{\gamma < \delta}$ (denoted by $\mathrm{cof}((y_\gamma)_{\gamma < \delta})$ is the least ordinal $\alpha$ such that there exists some $(x_\beta)_{\beta < \alpha}$  which is cofinal in $(y_\gamma)_{\gamma < \delta}$.

\begin{lemma}
\label{lem:cofdoteq}
If $(x_\beta)_{\beta < \alpha} \doteq (y_\gamma)_{\gamma < \delta}$, then there is some $(y'_\lambda)_{\lambda < \alpha'} \subseteq (y_\gamma)_{\gamma < \delta}$ with $\alpha' \leq \alpha$ and $(y'_\lambda)_{\lambda < \alpha'} \doteq (y_\gamma)_{\gamma < \delta}$.
\begin{proof}
For each $x_\beta$, let $\gamma_\beta = \min \{\gamma \mid x_\beta \preceq y_\gamma\}$. By assumption, the set on the right hand side is non-empty; and as it is a set of ordinals, it has a minimum. The set $\{y_{\gamma_\beta} \mid \beta < \alpha\}$ is well-ordered by $\preceq$ (as a subset of a well-ordered set). Hence it can be turned into an increasing chain $(y'_\lambda)_{\lambda < \alpha'}$. By construction, we have $(y'_\lambda)_{\lambda < \alpha'} \subseteq (y_\gamma)_{\gamma < \delta}$ and $(x_\beta)_{\beta < \alpha} \sqsubseteq (y'_\lambda)_{\lambda < \alpha'}$. By transitivity, it follows that $(y'_\lambda)_{\lambda < \alpha'} \doteq (y_\gamma)_{\gamma < \delta}$.

It remains to argue that $\alpha' \leq \alpha$. For that, consider the map $\lambda \mapsto \min \{\beta \mid y_{\gamma_\beta} = y'_\lambda\}$. This map is well-defined, injective and preserves $\leq$. Thus, it constitutes an order-homeomorphism from $\alpha'$ to $\alpha$, and witnesses that $\alpha' \leq \alpha$.
\end{proof}
\end{lemma}

\begin{corollary}
\label{corr:cofdoteq}
$\mathrm{cof}((y_\gamma)_{\gamma < \delta})$ is equal to the least ordinal $\alpha$ such that there exists $(x_\beta)_{\beta < \alpha}$ with $(x_\beta)_{\beta < \alpha} \doteq (y_\gamma)_{\gamma < \delta}$.
\begin{proof}
Since any increasing chain that is cofinal in $(y_\gamma)_{\gamma < \delta}$ is equivalent to $(y_\gamma)_{\gamma < \delta}$, it follows that $\mathrm{cof}((y_\gamma)_{\gamma < \delta})$ is an upper bound.

Conversely, if $(x_\beta)_{\beta < \alpha} \doteq (y_\gamma)_{\gamma < \delta}$, by Lemma \ref{lem:cofdoteq} there is some $(y'_\lambda)_{\lambda < \alpha'} \subseteq (y_\gamma)_{\gamma < \delta}$ with $\alpha' \leq \alpha$ and $(y'_\lambda)_{\lambda < \alpha'} \doteq (y_\gamma)_{\gamma < \delta}$. By the definition of cofinality, we have that $\alpha' \geq \mathrm{cof}((y_\gamma)_{\gamma < \delta})$, so in particular, also $\alpha \geq \mathrm{cof}((y_\gamma)_{\gamma < \delta})$.
\end{proof}
\end{corollary}

\begin{corollary}\label{corr:cofinality_countable}
For every chain $(y_\gamma)_{\gamma < \delta}$ there exists an equivalent chain $(x_\beta)_{\beta < \alpha}$ such that $\alpha = 1$ or $\alpha$ is an infinite regular cardinal.
In particular, if $\delta$ is countable, then $(y_\gamma)_{\gamma < \delta}$ is equivalent to a singleton or some chain $(x_n)_{n < \omega}$.
\end{corollary}

\begin{lemma}
\label{lemma:singletons}
If $(x_\beta)_{\beta < \alpha} \sqsubseteq (y_\gamma)_{\gamma < \delta}$ and $\alpha < \mathrm{cof}((y_\gamma)_{\gamma < \delta})$, then there exists $\gamma_0 < \delta$ such that
\[(x_\beta)_{\beta < \alpha} \sqsubseteq (y_{\gamma_i})_{i<1}\]
\end{lemma}
\begin{proof}
As in the proof of Lemma \ref{lem:cofdoteq}: For each $x_\beta$, pick some $\gamma_\beta$ with $x_\beta \preceq y_{\gamma_\beta}$. The set $\{y_{\gamma_\beta} \mid \beta < \alpha\}$ is well-ordered by $\preceq$ (as a subset of a well-ordered set), and has cardinality at most $|\alpha|$. Hence there is some $(y'_\beta)_{\beta < \alpha'}$ with $|\alpha'| \leq |\alpha|$, $(y'_\beta)_{\beta < \alpha'} \subseteq (y_\gamma)_{\gamma < \delta}$ and $(x_\beta)_{\beta < \alpha} \sqsubseteq (y'_\beta)_{\beta < \alpha'}$.

By assumption, $(y'_\beta)_{\beta < \alpha'}$ cannot be cofinal in $(y_\gamma)_{\gamma < \delta}$. Thus, there has to be some $\gamma_0 < \delta$ such that for no $\beta < \alpha'$ we have that $y_{\gamma_0} \preceq y'_\beta$. But as $(y_\gamma)_{\gamma < \delta}$ is totally ordered, this implies that for all $\beta < \alpha$ we have $y'_\beta \preceq y_{\gamma_0}$, i.e.~$(y'_\beta)_{\beta < \alpha'} \sqsubseteq (y_{\gamma_i})_{i<1}$. The claim follows by transitivity of $\sqsubseteq$.
\end{proof}

We briefly illustrate the concepts introduced so far in the game setting. Notice that for a game $\game$ and a Player~$i$, the pair  $(\Sigma_i(\game), \preceq)$ is indeed a quasi-ordered set.
We can thus consider the set $IC(\Sigma_i(\game), \preceq)$ of increasing chains of strategies in $\game$.

\begin{example}\label{ex:best_animal_2}
Recall the \emph{Help-me?} game of Figure~\ref{fig:best_animal} and consider the set $(\Sigma_i, \preceq)$ of strategies of the protagonist ordered by the weak dominance relation.
Any single strategy is an increasing chain, indexed by the ordinal $1$.
We already noted that the strategy $s_\omega$ is admissible, thus the chain consisting of $s_\omega$ is maximal with respect to $\sqsubseteq$.
Furthermore, the sequence of strategies $(s_k)_{k < \omega}$ is an increasing chain.
Indeed, we know that for any $k < \omega$, we have $s_k \prec s_{k+1}$.
It is a maximal one: in fact, since the set of strategies of the protagonist solely consists of the strategies of this chain and $s_\omega$, and as $s_k \not \preceq s_\omega$ for any $k < \omega$, we get that any chain $(\sigma_\beta)_{ \beta < \alpha}$ such that $(s_k)_{k < \omega}\sqsubseteq (\sigma_\beta)_{ \beta < \alpha}$ satisfies $(\sigma_\beta)_{ \beta < \alpha} \subseteq (s_k)_{k < \omega}$.
Thus, $(\sigma_\beta)_{ \beta < \alpha} \sqsubseteq (s_k)_{k < \omega}$.
Let $(\sigma_\beta)_{ \beta < \alpha}$ be an increasing chain indexed by the ordinal $\alpha$.
First, remark that $\alpha \leq \omega$.
If $\alpha < \omega$, then the cofinality of $(\sigma_\beta)_{ \beta < \alpha}$ is $1$ as $(\sigma_\beta)_{ \beta < \alpha}$ is equivalent to the strategy $\sigma_{\alpha -1}$.
If $\alpha = \omega$, then the cofinality of $(\sigma_\beta)_{ \beta < \alpha}$ is $\omega$:
As for every finite chain $(\sigma'_{\beta'})_{ \beta' < \alpha'}$ with $1< \alpha' < \omega$, there exists $n < \omega$ such that $(\sigma'_{\beta'})_{ \beta' < \alpha'} \sqsubset \sigma_n$, and thus $(\sigma_\beta)_{ \beta < \alpha}$ is not (weakly) dominated by $(\sigma'_{\beta'})_{ \beta' < \alpha'}$.
Moreover, we have that $(\sigma_\beta)_{ \beta < \alpha} \doteq  (s_k)_{k < \omega}$ and is thus maximal.
Indeed, since $(\sigma_\beta)_{ \beta < \alpha}$ is a chain that is not a singleton, we already know that $(\sigma_\beta)_{ \beta < \alpha} \subseteq (s_k)_{k < \omega}$, that is $(\sigma_\beta)_{ \beta < \alpha} \sqsubseteq (s_k)_{k < \omega}$.
Let now $k< \omega$.
As $(\sigma_\beta)_{ \beta < \alpha}$ is an increasing chain and $\alpha = \omega$, we have that there exists $n < \omega$ and $k' \geq k$ such that $\sigma_n = s_{k'}$.
Thus, $s_k \preceq \sigma_n$ since $(s_k)_{k < \omega}$ is an increasing chain.
Hence, we also have $(s_k)_{k < \omega} \sqsubseteq (\sigma_\beta)_{ \beta < \alpha}$.

\end{example}

Now we are ready to prove the main technical result of this section \ref{subsec:orderingchains}, which identifies the potential obstructions for each chain in $\mathrm{IC}$ to have an upper bound:

\begin{lemma}
\label{lemma:theobstruction}
The following are equivalent:
\begin{enumerate}
\item If $((x^\gamma_{\beta})_{\beta < \alpha_\gamma})_{\gamma < \delta}$ is an increasing chain in $\mathrm{IC}$, then it has an upper bound in $\mathrm{IC}$.
\item If $((x^\gamma_{\beta})_{\beta < \alpha})_{\gamma < \delta}$ is an increasing chain in $\mathrm{IC}$ with $\alpha \neq \delta$, $\mathrm{cof}((x^\gamma_{\beta})_{\beta < \alpha}) = \alpha > 1$ and $\mathrm{cof}(((x^\gamma_{\beta})_{\beta < \alpha})_{\gamma < \delta}) = \delta > 1$, then it has an upper bound in $\mathrm{IC}$.
\end{enumerate}
\end{lemma}
\begin{proof}
It is clear that $2$ is a special case of $1$. We thus just need to show that any potential obstruction to $1$ can be assumed to have the form in $2$.

By replacing each $(x^\gamma_{\beta})_{\beta < \alpha_\gamma}$ with some suitable cofinal increasing chain if necessary, we can assume that $\mathrm{cof}((x^\gamma_{\beta})_{\beta < \alpha_\gamma}) = \alpha_\gamma$ for all $\gamma < \delta$.

Consider $\{(x^\gamma_{\beta})_{\beta < \alpha_\gamma} \mid \exists \gamma' > \gamma \ \alpha_\gamma < \alpha_{\gamma'}\}$. If this set is cofinal in $((x^\gamma_{\beta})_{\beta < \alpha_\gamma})_{\gamma < \delta}$, then for each $\gamma$ inside that set pick some witness $\gamma'$, and let $y_{\gamma}$ be the witness obtained from Lemma \ref{lemma:singletons}. Now $\{y_\gamma \mid \exists \gamma' > \gamma \ \alpha_\gamma < \alpha_{\gamma'}\}$ is the desired upper bound.

If the set from the paragraph above is not cofinal, then there exists some $\delta' < \delta$ such that for $\delta' \leq \gamma < \gamma' < \delta$ we always have that $\alpha_\gamma \geq \alpha_{\gamma'}$. As the $\alpha_\gamma$ are ordinals, decreases can happen only finitely many times. Thus, by moving to a suitable cofinal subset we can safely assume that all $\alpha_\gamma$ are equal to some fixed $\alpha$.

Again by moving to a suitable cofinal subset, we can assume that $\mathrm{cof}(((x^\gamma_{\beta})_{\beta < \alpha})_{\gamma < \delta}) = \delta$. If $\delta = 1$, the statement is trivial. If $\alpha = 1$, then $(x^\gamma_0)_{\gamma < \delta}$ is the desired upper bound. It remains to handle the case $\alpha = \delta > 1$.

We construct some function $f : \alpha \to \alpha$, such that the desired upper bound $(y_\epsilon)_{\epsilon < \alpha}$ is of the form $y_\epsilon = x_{f(\epsilon)}^\epsilon$. We proceed as follows: Set $f(0) = 0$. Once $f(\zeta)$ has been defined for all $\zeta < \epsilon$, pick for each $\zeta < \epsilon$ some $g(\zeta)$ such that $x_{f(\zeta)}^\zeta \preceq x_{g(\zeta)}^\epsilon$ and $x_{\epsilon}^\zeta \preceq x_{g(\zeta)}^\epsilon$. As $\epsilon < \alpha$, it cannot be that $\{x_{g(\zeta)}^\epsilon \mid \zeta < \epsilon\}$ is cofinal in $\{x_\beta^\epsilon \mid \beta < \alpha\}$. Thus, it has some upper bound, and we define $f(\epsilon)$ such that $x_{f(\epsilon)}^\epsilon$ is such an upper bound.
\end{proof}

Let us illustrate the problem of extending Lemma \ref{lemma:theobstruction} by an example:

\begin{example}[{\cite[Example 1]{markowsky}}]
\label{example:omega1omega}
Let $(X,\preceq) = \omega_1 \times \omega$, i.e.~the product order of the first uncountable ordinal and the first infinite ordinal. Consider the chain of chains given by $x_n^\gamma = (\gamma,n)$, this corresponds to the case $\alpha = \omega$, $\delta = \omega_1$ in Lemma \ref{lemma:theobstruction}. If this chain of chains had an upper bound, then $\omega_1 \times \omega$ would need to admit a cofinal chain. However, this is not the case.
\end{example}

However, we can guarantee the existence of a maximal chain above any chain when there is no uncountable increasing chain of increasing chains.

\begin{theorem}
\label{theo:maximalelements}
If all increasing chains of elements in $\mathrm{IC}$ (i.e., increasing chains of increasing chains of elements of $(X, \preceq)$) have a countable number of elements, then for every $A \in \mathrm{IC}$ there exists a maximal $B \in \mathrm{IC}$ with $A \sqsubseteq B$.
%\Marie{rephrase}
\begin{proof}
We first argue that Condition 2 in Lemma \ref{lemma:theobstruction} is vacuously true. As all increasing chains in $\mathrm{IC}$ are countable, the only possible value $\delta > 1$ for $\delta = \mathrm{cof}(((x^\gamma_{\beta})_{\beta < \alpha})_{\gamma < \delta})$ is $\delta = \omega$. As $(X,\preceq)$ embeds into $\mathrm{IC}$, if all chains in $\mathrm{IC}$ are countable, then so are all chains in $(X,\preceq)$. This tells us that the only possible value for $\alpha$ is $\alpha = \omega$. But then $\alpha \neq \delta$ cannot be satisfied.

By Lemma \ref{lemma:theobstruction}, Condition 1 follows. We can then apply Zorn's Lemma (Lemma \ref{lemma:zorn}) to conclude the claim.
\end{proof}
\end{theorem}

A small modification of the example shows that we cannot replace the requirement that $\mathrm{IC}$ has only countable increasing chains in Theorem \ref{theo:maximalelements} with the simpler requirement that $(X,\preceq)$ has only countable increasing chains:
\begin{example}\label{example:countable_chains_not_sufficient}
Let $X = \omega_1 \times \omega$, and let $(\alpha,n) \prec (\beta,m)$ iff $\alpha \leq \beta$ and $n < m$. Then $(X,\preceq)$ has only countable increasing chains, but $\mathrm{IC}$ still has the chain of chains given by $x_n^\gamma = (\gamma,n)$ as in Example \ref{example:omega1omega}.

\end{example}

\subsection{Uncountably long chains of chains}

Unfortunately, we can design a game such that there exists an uncountable increasing chain of increasing chains. Thus the existence of a maximal element above any chain is not guaranteed by Theorem \ref{theo:maximalelements}. In fact, we will see that the chain of chains of uncountable length we construct is not below any maximal chain.

\begin{figure}[t]
%\centering
\begin{subfigure}[b]{0.4\textwidth}
    \begin{tikzpicture}[->,>=stealth',shorten >=1pt, initial text={}]
      \node [initial above,state,rectangle] (q5) {$v_0$};
      \node [state] (q1)    [below=of q5]                  {$v_1$};
      \node[state,rectangle]          (q2) [right=of q1]         {$v_2$};
      \node[state]          (q3) [right=of q2]         {$\ell_2$};
     \node[state]          (q4) [left=of q1]         {$\ell_1$};
      \path (q5) edge node {$$} (q1);
            \path (q5) edge [loop right] node {$$} (q5);
      \path (q1) edge [bend left] node {$$} (q2);
      \path (q2) edge [bend left] node {$$} (q1);
      \path (q2) edge node {$$} (q3);

      \path (q1) edge node {$$} (q4);
       \path (q3) edge [loop above] node {$$} (q3);
            \path (q4) edge [loop above] node {$$} (q4);
    \end{tikzpicture}
    \caption{A variant of the \emph{Help-me?} game with an extra loop}
    \label{fig:bestanimalplusloop}
    \end{subfigure}
    \hfil
    \begin{subfigure}[b]{0.4\textwidth}
    \begin{tikzpicture}[->,>=stealth',shorten >=1pt, initial text={}]

      \node [initial above, state] (q1)                     {$v_0$};
       \node [state,rectangle] [above right=of q1] (q5) {$v_1$};
      \node[state,rectangle]          (q2) [right=of q1]         {$v_2$};
      \node[state]          (q3) [right=of q2]         {$\ell_2$};
     \node[state]          (q4) [left=of q1]         {$\ell_1$};
      \path (q5) edge node {$ $} (q2);
      \path (q1) edge node[above left] {$ b $} (q5);

      \path (q1) edge [bend left] node[below] {$ a $} (q2);
      \path (q2) edge [bend left] node {$$} (q1);
      \path (q2) edge node {$$} (q3);

      \path (q1) edge node {$$} (q4);
       \path (q3) edge [loop above] node {$$} (q3);
            \path (q4) edge [loop above] node {$ $} (q4);
    \end{tikzpicture}
    \caption{A variant of the \emph{Help-me?} game with two paths from $v_0$ to $v_2$}
    \label{fig:bestanimal2choices_1}
\end{subfigure}
\caption{Variants of the \emph{Help-me?} game}
\end{figure}
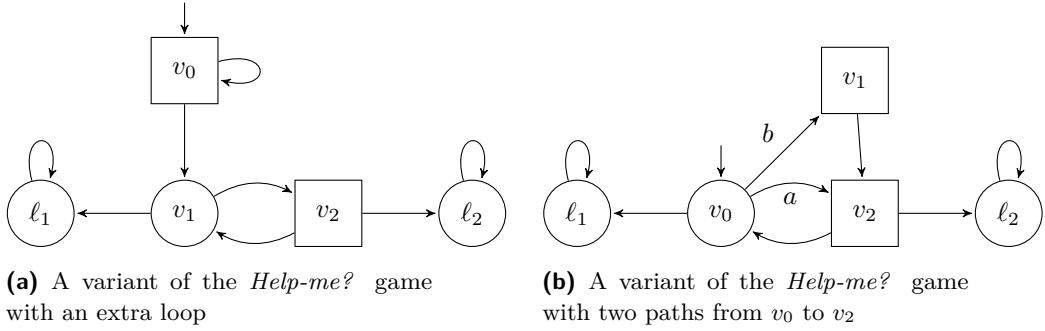

\begin{example}
\label{ex:bestanimalplusloop}
We consider a variant of the \emph{Help-me?} game (Example \ref{ex:best_animal}), depicted in Figure \ref{fig:bestanimalplusloop}.
The strategies of the protagonist in this game can be described by functions $f : \mathbb{N} \to \mathbb{N} \cup \{\infty\}$ describing how often the protagonist is willing to repeat the second loop (between $q_1$ and $q_2$) given the number of repetitions the antagonist made in the first loop (at $q_0$). With the same reasoning as in Example \ref{ex:best_animal} we find that the strategy corresponding to a function $g$ dominates the strategy corresponding to $f$ iff $\forall n \in \mathbb{N} \ f(n) = \infty \Leftrightarrow g(n) = \infty$ and $\forall n \in  \mathbb{N} f(n) \leq g(n)$.
\end{example}

\begin{definition}
Let $\mathbb{N}^\mathbb{N}$ denote the set of functions $f : \mathbb{N} \to \mathbb{N}$. For $f, g \in \mathbb{N}^\mathbb{N}$, let $f \leq g$ denote that $\forall n \in \mathbb{N} \ f(n) \leq g(n)$.
\end{definition}

\begin{observation}
\label{obs:forchains}
There is an embedding of $(\mathbb{N}^\mathbb{N},\leq)$ into the strategies of the game in Example \ref{ex:bestanimalplusloop} ordered by dominance such that no strategy in the range of embedding is dominated by a strategy outside the range of the embedding.
\end{observation}

\begin{proposition}[\footnote{This result is adapted from an answer by user \emph{Deedlit} on \url{math.stackexchange.org} \cite{stackexchange}.}]
\label{prop:longchains}
For every chain $(f_n)_{n \in \mathbb{N}}$ in $(\mathbb{N}^\mathbb{N},\leq)$ there exists a chain of chains $((f_n^\alpha)_{n < \omega})_{\alpha < \omega_1}$ of length $\omega_1$ with $(f_n^0)_{n < \omega} \sqsupseteq (f_n)_{n < \omega}$.
\end{proposition}
\begin{proof}
For each countable limit ordinal $\alpha$, we fix\footnote{We have no computability or other uniformity requirements to satisfy, and can thus just invoke the axiom of choice. Otherwise, as discussed e.g.~in \cite[Section 3.1]{promel} this approach would fail.} some fundamental sequence $(\alpha[m])_{m < \omega}$ of ordinals with $\alpha[m] < \alpha$ and $\sup_{m \in \omega} \alpha[m] = \alpha$.

Let $f^0_n(k) = \max \{f(k), k\}$. Let $f_n^{\alpha + 1}(k) = \max_{j \leq k}(f_{n+j}^\alpha)(k) + 1$, and for limit ordinals $\alpha$, let $f_n^\alpha(k) = \max_{m \leq n + k} f_n^{\alpha[m]}(k)$.

{\bf Claim:} If $\alpha \leq \beta$, then $(f_n^\alpha)_{n < \omega} \sqsubseteq (f_m^\beta)_{m < \omega}$.
\begin{proof}
It suffices to show that if $\alpha \leq \beta$, then $f_n^\alpha \leq f_n^\beta$ for all $n$ greater than some $t$. If $\beta = \alpha + 1$, this is immediate already for $t = 0$. For $\beta$ a limit ordinal, we note that $f_n^{\beta[m]} \leq f_n^\beta$ for $n \geq m$.

The claim then follows by induction over $\beta$. Recall that if $\beta$ is a limit ordinal and $\alpha < \beta$, then there is some $m \in \omega$ with $\alpha \leq \beta[m]$. Since for any given $\alpha,\beta$, the ordinals $\gamma$ between $\alpha$ and $\beta$ we will need to inspect in the induction form a decreasing chain, there are only finitely many such ordinals. In particular, the maximum of all thresholds $t$ we encounter is well-defined.
\end{proof}

{\bf Claim:} If $\alpha > \beta$, then $(f_n^\alpha)_{n < \omega} \not\sqsubseteq (f_m^\beta)_{m < \omega}$.
\begin{proof}
Due to transitivity of $\sqsubseteq$ and the previous claim, it suffices to show that $(f_m^{\alpha+1})_{m < \omega} \not\sqsubseteq (f_n^\alpha)_{n < \omega}$. Write $g_n = f_n^\alpha$. Assume the contrary, i.e.~that for all $n < \omega$ there exists some $m < \omega$ such that for all $k \in \mathbb{N}$ and for all $j \leq k$ we have that $g_{n+j}(k) + 1\leq g_m(k)$. In particular, for $n = 0$ we would have that $\forall k \in \mathbb{N} \ \forall j \leq k \ g_j(k) + 1\leq g_m(k)$, and then setting $k = j = m$, that $g_m(m) + 1 \leq g_m(m)$, which is a contradiction.
\end{proof}
\end{proof}

\begin{corollary}\label{cor:example_uncountable_chains}
The game in Example \ref{ex:bestanimalplusloop} has uncountably long chains of chains not below any maximal chains.
\begin{proof}
Combine Observation \ref{obs:forchains} and Proposition \ref{prop:longchains}.
\end{proof}
\end{corollary}

\subsection{Chains over countable quasiorders $(X,\preceq)$}
Our proof of Proposition \ref{prop:longchains} crucially relied on functions of type $f : \mathbb{N} \to \mathbb{N}$ with arbitrarily high rate of growth. In concrete applications such functions would typically be unwelcome. In fact, for almost all classes of games of interest in (theoretical) computer science, a countable collection of strategies suffices for the players to attain their attainable goals. Restricting to computable strategies often makes sense. Many games played on finite graphs are even finite-memory determined (see \cite{paulyleroux4} for how this extends to the quantitative case), and thus strategies implementable by finite automata are all that need to be considered.

Restricting consideration to a countable set of strategies indeed circumvents the obstacle presented by Proposition \ref{prop:longchains}. The reason is that the cardinality of the length of a chain of chains cannot exceed that of the underlying quasiorder $(X,\preceq)$:

\begin{proposition}
\label{prop:noverylongchains}
For any increasing chain $((x^\gamma_{\beta})_{\beta < \alpha})_{\gamma < \delta}$ in $\mathrm{IC}(X,\preceq)$ we find that $|\delta| \leq |X|$.
\begin{proof}
Let $X_\gamma = \{x \in X \mid \exists \beta < \alpha \ x \preceq x^\gamma_\beta\}$. We find that $X_{\gamma_1} \subsetneq X_{\gamma_2}$ for any $\gamma_1 < \gamma_2 < \delta$ as a direct consequence of $(x^{\gamma_1}_{\beta})_{\beta < \alpha} \sqsubset (x^{\gamma_2}_{\beta})_{\beta < \alpha}$. Pick for each $\gamma < \delta$ some $y_\gamma \in X_{\gamma + 1} \setminus X_{\gamma}$. Then $y_{\cdot} : \delta \to X$ is an injection, establishing $|\delta| \leq |X|$.
\end{proof}
\end{proposition}

\begin{corollary}
\label{corr:countablymanystrategiesgood}
If $(X,\preceq)$ is countable, then any increasing chain is maximal or below a maximal chain.
\begin{proof}
Proposition \ref{prop:noverylongchains} shows that Theorem \ref{theo:maximalelements} applies.
\end{proof}
\end{corollary}

 \begin{example}\label{ex:best_animal_3}
We return to the \emph{Help-me?} game (Example \ref{ex:best_animal}, Figure \ref{fig:best_animal}). With the analysis done in Example~\ref{ex:best_animal_2}, we have seen that any increasing chain $C$ is either maximal or such that $C \sqsubseteq (\sigma_n)_{n < \omega}$, which is maximal.
This fact can be derived directly from Corollary \ref{corr:countablymanystrategiesgood} as the number of strategies in $\game$ is countable. Note also that the seemingly irrelevant loop we added in Figure \ref{fig:bestanimalplusloop} has a fundamental impact on the behaviour of chains of strategies!
\end{example}

\section{Generalised safety/reachability games}

\begin{definition}\label{def:gen_safety_reach_game}
A \emph{generalised safety/reachability game} (for Player~$i$) $\game=\langle P,G, L, (p_i)_{i\in P}\rangle$ is a turn-based multiplayer game on a finite graph such that:
\begin{itemize}[topsep=0pt]
\item $L \subseteq V$ is a finite set of \emph{leaves},
\item for each $\ell \in L$, we have that $(\ell, v) \in E$ if, and only if $v=\ell$, that is, each leaf is equipped with a self-loop, and no other outgoing transition,
\item for each $\ell \in L$, there exists an associated payoff $n_{\ell} \in \Z$ such that:
for each outcome $\rho$, we have $p_i(\rho) = \begin{cases}
n_{\ell} \text{~if~} \rho \in V^*{\ell}^{\omega},\\
0 \text{~otherwise.~ }
\end{cases}
$
\end{itemize}
\end{definition}

The traditional reachability games can be recovered as the special case where all leaves are associated with the same positive payoff, whereas the traditional safety games are those generalised safety/reachability games with a single negative payoff attached to leaves. This class was studied under the name \emph{chess-like} games in \cite{boros,boros2}.

Generalised safety/reachability games are \emph{well-formed} for Player~$i$.
Furthermore, they are prefix-independent, that is, for any outcome $\rho$ and history $h$, we have that $p_i(h\rho) = p_i(\rho)$.
Without loss of generality, we consider that there is either a unique leaf $\ell(n) \in L$ or no leaf for each possible payoff $n \in \Z$.

It follows from the transfer theorem in \cite{paulyleroux4} (in fact, already from the weaker transfer theorem in \cite{depril3}) that generalised safety/reachability games are finite memory determined.
With a slight modification, we see that for any  history $h$ and strategy $\sigma$, there exists a finite-memory strategy $\sigma'$ such that $\cVal(h,\sigma')=\cVal(h,\sigma)$ and $\aVal(h,\sigma')=\aVal(h,\sigma)$. We shall thus restrict our attention to finite memory strategies, of which there are only countably many. We then obtain immediately from Corollary \ref{corr:countablymanystrategiesgood}:

\begin{corollary}
In a generalised safety/reachability game, every increasing chain comprised of finite memory strategies is either maximal or dominated by a maximal such chain.
\end{corollary}

If our goal is only to obtain a dominance-related notion of rationality, then for generalised safety/reachability games we can be satisfied with maximal chains comprised of finite memory strategies. However, for applications, it would be desirable to have a concrete understanding of these maximal chains. For this, having used Zorn's Lemma in the proof of their existence surely is a bad omen!

After collecting some useful lemmas on dominance in generalised safety/reachability games in Section \ref{subsec:dominancelemmas}, we will introduce the notion of \emph{uniform chains} in Section \ref{subsec:uniformchains}. These are realized by automata of a certain kind, and thus sufficiently concrete to be amenable to algorithmic manipulations.

\subsection{Dominance in generalised safety/reachability games}
\label{subsec:dominancelemmas}
Given a generalised safety/reachability game $\game$ and two strategies $\sigma_1$ and $\sigma_2$ of Player~$i$, we can provide a criterion to show that $\sigma_1$ is not dominated by $\sigma_2$:

\begin{lemma}\label{lem:dominance_histories}
Let $\sigma_1$ and $\sigma_2$ be two strategies of Player $i$ in a generalised safety/reachability game $\game$.
Then, $\sigma_1 \not \preceq \sigma_2$ if, and only if,
there exists a history $h$ compatible with $\sigma_1$ and $\sigma_2$ such that $\last h \in V_i$, $\sigma_1(h) \neq \sigma_2(h)$ and $\cVal(h,\sigma_1) > \aVal(h,\sigma_2)$.
\end{lemma}

Intuitively, if there is no history where the two strategies disagree, they are in fact equivalent, and if, at every history where they disagree, the best payoff $\sigma_1$ can achieve (that is, $\cVal(h,\sigma_1)$) is less than the one $\sigma_2$ can ensure (that is, $\aVal(h,\sigma_2)$), then $\sigma_1 \preceq \sigma_2$.
On the other hand, if they disagree at a history $h$ and the best payoff $\sigma_1$ can achieve is strictly greater than the one $\sigma_2$ can ensure, then there exist a strategy of the antagonist that will yield exactly these payoffs against $\sigma_1$ and $\sigma_2$ respectively, which means that $\sigma_1 \not \preceq \sigma_2$.
This result follows from the proof of Theorem~$11$ in~\cite{BrenguierPRS16}.
The proof adapted to our setting can be found in the appendix. %\Marie{ref at~\cite{DBLP:conf/fsttcs/BrenguierPRS16}}

\begin{proof}[Proof of Lemma~\ref{lem:dominance_histories}]
\begin{itemize}
\item[$\Longrightarrow$] Suppose that for every history $h$ compatible with $\sigma_1$ and $\sigma_2$ such that $\last h  \in V_i$ and $\sigma_1(h) \neq \sigma_2(h)$, we have that $\cVal(h,\sigma_1) \leq \aVal(h,\sigma_2)$.
We show that $\sigma_1 \preceq \sigma_2$.
Let $\tau$ be a strategy of Player $-i$.
Consider $\rho_1 = \outcome( \sigma_1, \tau)$ and $\rho_2 = \outcome(\sigma_2, \tau)$.
If for all prefixes $h' \prefix \rho_1$ such that $\last {h'} \in V_i$, it holds that $\sigma_1({h'}) = \sigma_2(h')$, then in fact $\rho_1 = \rho_2$ and $p_i(\sigma_1, \tau) = p_i(\sigma_2, \tau)$.
Otherwise, let $h$ be the least common prefix of $\rho_1$ and $\rho_2$ such that $\last h  \in V_i$ and $\sigma_1(h) \neq \sigma_2(h)$.
We know that $p_i(\rho_1) \leq \cVal(h, \sigma_1)$ and $p_i(\rho_2) \geq \aVal(h, \sigma_2)$ since $h \prefix \rho_1$ and $h \prefix \rho_2$.
As $\cVal(h,\sigma_1) \leq \aVal(h,\sigma_2)$, we have that $p_i(\sigma_1, \tau) \leq p_i(\sigma_2, \tau)$.
Thus, for every $\tau \in \Sigma_{-i}$, it holds that $p_i(\sigma_1, \tau) \leq p_i(\sigma_2, \tau)$, that is, $\sigma_1 \preceq \sigma_2$.

\item[$\Longleftarrow$] Let $h$ be a history compatible with $\sigma_1$ and $\sigma_2$ such that $\last h \in V_i$, $\sigma_1(h) \neq \sigma_2(h)$ and $\cVal(h,\sigma_1) > \aVal(h,\sigma_2)$.
Then, there exists two strategies $\tau_1$ and $\tau_2$ of player $-i$ such that $p_i(h,\sigma_1, \tau_1) = \cVal(h,\sigma_1)$ and $p_i(h,\sigma_2, \tau_2) = \aVal(h,\sigma_2)$.
Let $\tau$ be a strategy of player $-i$ compatible with $h$, and define $\tau' (h') = \begin{cases}
\tau_1(h')   \text{~if~} h\sigma_1(h) \prefix h', \\
\tau_2(h')   \text{~if~} h\sigma_2(h) \prefix h', \\
\tau(h') \text{~otherwise.~}
\end{cases}
$

The strategy $\tau'$ is well defined, as $\sigma_1(h) \neq \sigma_2(h)$.
Furthermore, we have that $p_i(\sigma_1, \tau') = p_i(h,\sigma_1, \tau_1) = \cVal(h,\sigma_1) > \aVal(h,\sigma_2) = p_i(h,\sigma_2, \tau_2) = p_i(\sigma_2, \tau')$, since generalised safety/reachability games are prefix-independent.
Thus, $\sigma_1 \not \preceq \sigma_2$.
\end{itemize}

\end{proof}

\medskip

We call such a history $h$ a \emph{non-dominance witness} of $\sigma_1$ by $\sigma_2$. The existence of non-dominance witnesses allows us to conclude that in generalised safety/reachability games, all increasing chains are countable (not just those comprised of finite memory strategies).

\begin{corollary}
If $(\sigma_\beta)_{\beta < \alpha}$ is an increasing chain in generalised safety/reachability game, then $\alpha$ is countable.
\begin{proof}
Assume that a history $h$ is a witness of non-dominance of $\sigma_2$ by $\sigma_1$, and of $\sigma_3$ by $\sigma_2$, but not of $\sigma_1$ by $\sigma_2$ or $\sigma_2$ by $\sigma_3$. Then $\cVal(h,\sigma_2) > \aVal(h,\sigma_1)$, $\cVal(h,\sigma_3) > \aVal(h,\sigma_2)$, $\cVal(h,\sigma_1) \leq \aVal(h,\sigma_2)$ and $\cVal(h,\sigma_2) \leq \aVal(h,\sigma_3)$. It follows that $\aVal(h,\sigma_1) < \aVal(h,\sigma_3)$ and $\cVal(h,\sigma_1) < \cVal(h,\sigma_3)$. Thus, if there are $k$ different possible values, then any increasing chain of strategies using $h$ as witness of non-dominance between them can have length at most $2k - 1$.

But if there were an uncountably long increasing chain, by the pigeon hole principle it would have an uncountably long subchain where all non-dominance witnesses in the reverse direction are given by the same history.
\end{proof}
\end{corollary}

As we only handle countable chains, in the following we use the usual notation $(\sigma_n)_{n \in \N}$ to index chains.

The following lemma states that we can also extract witnesses for a strategy to be non-maximal (non-admissible or strictly dominated):

\begin{lemma}\label{lem:strict_dominance}
Let $\game$ be a generalised safety/reachability game and $\sigma$ a strategy of Player $i$.
The strategy $\sigma$ is not admissible if, and only if there exists a history $h$ compatible with $\sigma$ such that $\aVal(h,\sigma) \leq \cVal(h, \sigma) \leq \aVal(h) \leq acVal(h)$ where at least one inequality is strict.
\end{lemma}

This result is a reformulation of Theorem~$11$ in~\cite{BrenguierPRS16} catered to our context and with a focus on the non-admissibility rather than on admissibility. For the proof, we use the following notation: Given two strategies $\sigma, \sigma'$ of Player~$i$, and a history $h$, $\sigma[h \leftarrow \sigma']$ denotes the strategy that follows $\sigma$ and \emph{shifts} to $\sigma'$ at history $h$.
Formally, for all histories $h' \in \history(\game)$, the strategy $\sigma[h \leftarrow \sigma']$ is such that:
\[ \sigma[h \leftarrow \sigma'](h') =
\begin{cases}
\sigma'(h') \text{~if~} h \prefix h', \\
\sigma(h') \text{~otherwise~}.

\end{cases}
\]

\begin{proof}[Proof of Lemma \ref{lem:strict_dominance}]
\begin{itemize}
\item[$\Leftarrow$] Let $\sigma'_h$ be a strategy compatible with $h$ such that $\aVal(h,\sigma'_h) = \aVal(h)$ and $\cVal(h,\sigma'_h) = \acVal(h)$.
Let $\sigma'$ be the strategy that follows $\sigma$ and switches to $\sigma'_h$ at history $h$, that is, $\sigma'= \sigma[h \rightarrow \sigma'_h]$.
We claim that $\sigma \prec \sigma'$.
Indeed, we know that $\aVal(h,\sigma) \leq \cVal(h, \sigma) \leq \aVal(h) \leq \acVal(h)$.
For every strategy $\tau \in \Sigma_{-i}$, if $h \not \prefix \outcome {\sigma, \tau}$, by definition of $\sigma'$, we have that $\outcome{\sigma', \tau} = \outcome{\sigma,\tau}$, thus $p_i(\sigma, \tau) = p_i(\sigma', \tau)$.
If $h  \prefix \outcome {\sigma, \tau}$, then $h  \prefix \outcome {\sigma', \tau}$ and $\outcome{\sigma', \tau} = \outcome{h,\sigma'_h,\tau}$.
Hence $\aVal(h,\sigma) \leq p_i(\sigma,\tau) \leq \cVal(h,\sigma) \leq \aVal(h,\sigma'_h)= \aVal(h,\sigma') \leq p_i(\sigma',\tau) \leq \cVal(h,\sigma'_h) = \cVal(\sigma',\tau)$ .
Thus, $\sigma \preceq \sigma'$.
Furthermore, there is a strict inequality in $\aVal(h,\sigma) \leq \cVal(h, \sigma) \leq \aVal(h) \leq \acVal(h)$.
Hence, there exists $\tau$ such that $h \prefix \outcome {\sigma, \tau}$ and $p_i(\sigma,\tau) <p_i(\sigma',\tau)$.
\item[$\Rightarrow$]
Let $\sigma'$ be such that $\sigma \prec \sigma'$.
In particular, $\sigma \preceq \sigma'$.
By Lemma~\ref{lem:dominance_histories}, we know that there exists a history $h$ compatible with $\sigma$ and $\sigma'$ such that $\last h \in V_i$, $\sigma(h) \neq \sigma(h)$ and $\cVal(h,\sigma) \leq \aVal(h,\sigma')$.
Since the domination is strict between $\sigma$ and $\sigma'$, we further know that the sequence of inequalities $\aVal(h,\sigma)\leq \cVal(h,\sigma) \leq \aVal(h,\sigma')\leq \cVal(h,\sigma')$  with at least one strict inequality holds.
Towards contradiction, we assume that the chain of inequalities $\aVal(h,\sigma) \leq \cVal(h, \sigma) \leq \aVal(h) \leq \acVal(h)$ where at least one inequality is strict does not hold.
That is, either $\aVal(h,\sigma) = \cVal(h, \sigma) = \aVal(h) = \acVal(h)$ or $\cVal(h,\sigma) > \aVal(h)$.
Suppose that $\aVal(h,\sigma) = \cVal(h, \sigma) = \aVal(h) = \acVal(h)$.
As $\aVal(h,\sigma') \leq \aVal(h)$, we have that $\aVal(h,\sigma)= \cVal(h,\sigma) = \aVal(h,\sigma') < \cVal(h,\sigma')$.
Since $\sigma'$ guarantees $\aVal(h)$, we know that $\cVal(h,\sigma') \leq \acVal(h)$.
Thus $\acVal(h)=\aVal(h) < \cVal(h,\sigma') \leq \acVal(h)$, which is a contradiction.
Suppose now that $\cVal(h,\sigma) > \aVal(h)$.
As $\aVal(h,\sigma') \leq \aVal(h)$, it follows that $\cVal(h,\sigma) > \aVal(h, \sigma')$, which contradicts the fact that $\cVal(h,\sigma) \leq \aVal(h,\sigma')$.

\end{itemize}
\end{proof}

\begin{definition}
\label{def:preadmissible}
Call a history $h$ as in Lemma \ref{lem:strict_dominance} a non-admissibility witness for $\sigma$. Call $\sigma$ preadmissible, if for every non-admissibility witness $hv$ of $\sigma$ we find that $h = h'vh''$ with $\aVal(h'v,\sigma) = \aVal(h'v)$ and $\cVal(h'v,\sigma) = acVal(h'v)$.
\end{definition}

While a preadmissible strategy may fail to be admissible, it is not possible to improve upon it the first time it enters some vertex. Only when returning to a vertex later it may make suboptimal choices. Moreover, before a dominated choice is possible at a vertex, previously both the antagonistic and the antagonistic-cooperative value were realized at that vertex by the preadmissible strategy.

\begin{lemma}
\label{lemma:belowpreadmissible}
In a generalised safety/reachability game, every strategy is either preadmissible or dominated by a preadmissible strategy.
\begin{proof}
For each vertex $v$ in the game, we fix a finite memory strategy $\tau^v$ that realizes $\aVal(v)$ and $acVal(v)$. Note that since generalised safety/reachability games are prefix independent, values depend only on the current vertex, but not on the entire history.

We start with a finite memory strategy $\sigma$. If it is not already preadmissible, then it has witnesses of non-admissibility violating the desired property. Whether a history $h$ is a witness of non-admissibility for a finite memory strategy $\sigma$ depends only on the last vertex of $h$ and the current state of $\sigma$. We now modify $\sigma$ such that whenever $\sigma$ is in a combination of vertex $v$ and state $s$ corresponding to a problematic witness of non-admissibility, the new strategy $\sigma'$ moves to playing $\tau^v$ instead. The choices of $v$, $s$ and $\tau^k$ ensure that $\sigma'$ dominates $\sigma$.

The new strategy $\sigma'$ may fail to be preadmissible, again, and we repeat the construction. Now any problematic history in $\sigma'$ needs to enter the automaton for some $\tau^v$ at some point. By choice of $\tau^v$, the history where $\tau^v$ has just been entered cannot be a witness of non-admissibility. It follows that a problematic history entering $\tau^v$ cannot end in $v$. Repeating the updating process for at most as many times as there are vertices in the game graph will yield a preadmissible finite memory strategy dominating $\sigma$.
\end{proof}
\end{lemma}

\begin{lemma}
\label{lemma:nonadmissibilitynondominance}
If $h$ is not a witness of non-admissibility of $\sigma$, and not a witness of non-dominance of $\sigma$ by $\tau$, then $h$ is not a witness of non-dominance of $\tau$ by $\sigma$.
\end{lemma}

\begin{proof}[Proof of Lemma \ref{lemma:nonadmissibilitynondominance}]
To even be a candidate for a witness of non-dominance of $\tau$ by $\sigma$, $h$ needs to be compatible with $\sigma$ and $\tau$ and satisfy $\last h \in V_1$, $\sigma(h) \neq \tau(h)$. Not being a witness of non-dominance of $\sigma$ by $\tau$ then implies $\cVal(h,\sigma) \leq \aVal(h,\tau)$. It follows in particular that $\aVal(h,\sigma) \leq \cVal(h,\sigma) \leq \aVal(h,\tau) \leq \aVal(h) \leq \acVal(h)$. The only way $h$ can not be a witness of non-admissibility of $\sigma$ is if $\aVal(h,\sigma) = \cVal(h,\sigma) = \aVal(h,\tau) = \aVal(h) = \acVal(h)$. Since $\aVal(h,\tau) = \aVal(h)$, it follows that $\cVal(h,\tau) \leq \acVal(h,\tau) = \aVal(\sigma)$, i.e.~$h$ is not a witness of non-dominance of $\tau$ by $\sigma$.
\end{proof}

\begin{lemma}
\label{lemma:equalvalues}
Given an initialized game with initial vertex $v_0$, the following holds:
If for two strategies $\sigma$ and $\tau$ it holds that for any maximal history $h$ compatible with both, there is a prefix $h'$ with $\aVal(h',\sigma) = \aVal(h',\tau)$ and $\cVal(h',\sigma) = \cVal(h',\tau)$, then $\aVal(v_0,\sigma) = \aVal(v_0,\tau)$ and $\cVal(v_0,\sigma) = \cVal(v_0,\tau)$.
\end{lemma}
\begin{proof}
Assume for the sake of a contradiction that $\aVal(v_0,\sigma) < \aVal(v_0,\tau)$. Then there is a real $r \in \R$, and a strategy $\pi$ of the antagonist such that for any strategy $\pi'$ of the antagonist $p(\outcome(\sigma,\pi)) < r \leq p(\outcome(\tau,\pi'))$. If $\outcome(\sigma,\pi) = \outcome(\tau,\pi)$, this is clearly impossible. Thus, $\outcome(\sigma,\pi)$ and $\outcome(\tau,\pi)$ have some longest common prefix $h$, which is a maximal history compatible with $\sigma$ and $\tau$ (for it must be the protagonist who behaves differently first in $\outcome(\sigma,\pi)$ and $\outcome(\tau,\pi)$).

By assumption, $h$ has a prefix $h'$ with $\aVal(h',\sigma) = \aVal(h',\tau)$. Now $\aVal(h',\sigma) \leq p(\outcome(\sigma,\pi)) < r$. If $\aVal(h',\tau) < r$, then the antagonist must have a strategy $\pi''$ such that $h'$ is a prefix of $\outcome(\tau,\pi'')$ and $p(\outcome(\tau,\pi'')) < r$. But that contradictions $r \leq p(\outcome(\tau,\pi'))$ holding for all $\pi'$.

The proof for the cooperative value in place of the antagonistic one proceeds analogously.
\end{proof}

\begin{lemma}
\label{lemma:valuescoincide}Given an initialized game with initial vertex $v_0$, the following holds:
If $\sigma$ is preadmissible and $\sigma \preceq \tau$, then $\aVal(v_0,\sigma) = \aVal(v_0,\tau)$ and $\cVal(v_0,\sigma) = \cVal(v_0,\tau)$.
\end{lemma}

\begin{proof}
We show that the conditions of Lemma \ref{lemma:equalvalues} are satisfied, which will imply our desired conclusion. Consider a maximal history $h$ compatible with both $\sigma$ and $\tau$. First, assume that $h$ is not a witness of non-admissibility of $\sigma$. Since $\sigma \preceq \tau$, by Lemma \ref{lem:dominance_histories} $h$ cannot be a witness of non-dominance of $\sigma$ by $\tau$, i.e.~$\cVal(h,\sigma) \leq \aVal(h,\tau)$. By Lemma \ref{lemma:nonadmissibilitynondominance}, it follows that $h$ is not a witness of non-dominance of $\tau$ by $\sigma$ either, i.e.~$\cVal(h,\tau) \leq \aVal(h,\sigma)$. Put together, we have $\aVal(h,\sigma) = \cVal(h,\sigma) = \aVal(h,\tau) = \cVal(h,\tau)$.

It remains the case where $h$ is a witness of non-admissibility of $\sigma$. Then by preadmissibility of $\sigma$, $h$ has some prefix $h'$ with $\aVal(h',\sigma) = \aVal(h')$ and $\cVal(h',\sigma) = \acVal(h')$. Since $\sigma \preceq \tau$, we must have $\aVal(h',\sigma) \leq \aVal(h',\tau)$, so it follows that $\aVal(h',\sigma) = \aVal(h',\tau)$, and then that $\cVal(h',\tau) \leq \acVal(h') = \cVal(h',\sigma) \leq \cVal(h',\tau)$, i.e.~$\cVal(h',\sigma) = \cVal(h',\tau)$.
\end{proof}

\subsection{Parameterized automata and uniform chains}
\label{subsec:uniformchains}

Let a \emph{parameterized automaton} be a  Mealy automaton that in addition can access
a single counter in the following way:
In a counter-access-state, a transition is chosen based on
whether the counter value is $0$ or not.
Otherwise, the counter is decremented by $1$.

\begin{definition}\label{def:parameterized_automaton}
A \emph{parameterized automaton} for Player $i\in P$ over a game graph $G=(V,E)$ is a tuple $\mathcal{M} = (M, M_C, m_0, V, \mu, \nu)$ where:
\begin{itemize}[topsep=0pt]
\item $M$ is a non-empty finite set of memory states and $M_C \subseteq M$ is the set of \emph{counter-access} states,
\item $m_0$ is the initial memory state,
\item $V$ is the set of vertices of $G$,
%\item $c: M \to \lbrace 0,1\rbrace$ is the \emph{counter-state} function,
\item $\mu : M \times V \times \N \to M \times \N$ is the memory and counter update function,
\item $\nu : M \times V_i \times \N \to V$ is the move choice function for Player $i$, such that $(v, \nu(m,v,k)) \in E$ for all $m \in M$ and $v \in V_i$ and $n \in \N$.
\end{itemize}

The memory and counter-update function $\mu$ respects the following conditions:
for each $m \in M \setminus M_C$, and $v \in V$, there exists $m' \in M$ such that $\mu(m,v,n)=(m',n)$ for all $n \in \N$.
for each $m \in M_C$, and $v \in V$, there exists $m' \in M$ such that $\mu(m,v,n)=(m',n-1)$ for all $n > 0$ and $m'' \in M$ such that $\mu(m,v,0)=(m'',0)$.
The move choice function $\nu$ respects the following conditions:
for each $m \in M \setminus M_C$, and $v \in V_i$, there exists $v' \in V$ such that $\nu(m,v,n)=v'$ for all $n \in \N$.
for each $m \in M_C$, and $v \in V_i$, there exists $m' \in M$ such that $\nu(m,v,n)=(m',n)$ for all $n > 0$ and $m'' \in M$ such that $\nu(m,v,0)=(m'',0)$.
\end{definition}

To ease presentation and understanding, we call transitions that decrement the counter \emph{green} transitions, the transitions only taken when the counter value is $0$  \emph{red} transitions, and the ones that do not depend on the counter value \emph{black} transitions.
This classification between \emph{green}, \emph{red} and \emph{black} transitions extends naturally to the edges of the product $\mathcal{M} \times G$ (that is, the graph with set of vertices $M \times V$ and edges induced by the functions $\mu$ and $\nu$).

Parameterized automata can be seen as a collection of finite Mealy automata, one for each initialization of the counter.
Thus, we say that a parameterized automaton $\mathcal{M}$ realizes a \emph{sequence} of finite-memory strategies $(\sigma_n)_{n \in \N}$.
In the remainder of the paper, we focus on chains realized by parameterized automata:
%When this sequence is a chain, we call it \emph{uniform}.

\begin{definition}\label{def:uniform chain}
Let a chain $(\sigma_n)_{n\in \N}$ of regular strategies be called a \emph{uniform} chain if there is a parameterized
automaton M that realizes $\sigma_n$ if the counter is initialized with the value $n$. If $(\sigma_n)_{n\in \N}$ is maximal for $\sqsubseteq$ amongst the increasing chains comprised of finite memory strategies, we call it a a maximal uniform chain.
\end{definition}

\begin{example}
The \emph{Help-me?} game from Figure~\ref{fig:best_animal} is clearly a generalised safety/reachability game with two leaves.
The chain of strategies $(s_k)_{k \in \N}$ exposed in Example~\ref{ex:best_animal} is a uniform chain, as it is realized by the parameterized automaton that loops $k$ times when its counter is initialized with value $k$.
%Similarly, the chain of strategies $(s^2_k)_{k \in \N}$ where each $s^2_k$ loops $2k$ times before
Figure~\ref{fig:bestanimal_product} shows the product between this parameterized automaton and the game graph.
The green edge corresponds to the transition to take when the counter value is greater than $0$ and should be decremented, while the red edge corresponds to the transition to take when the counter value is  $0$.

\end{example}

\begin{figure}[h!]
\centering
    \begin{tikzpicture}[->,>=stealth',shorten >=1pt, initial text={}]
      \node [initial above ,state] (q1)                      {$q_0$};
      \node[state,rectangle]          (q2) [right=of q1]         {$q_1$};
      \node[state]          (q3) [right=of q2]         {$\ell(2)$};
     \node[state]          (q4) [left=of q1]         {$\ell(1)$};
      \path (q1) edge [bend left, green] node {$ $} (q2);
      \path (q2) edge [bend left] node {$$} (q1);
      \path (q2) edge node {$$} (q3);
       %\path (q3) edge node {$(-,Yes)$} (q4);
       %\path (q3) edge [bend left] node {$(-,No)$} (q1);
      \path (q1) edge [red] node {$$} (q4);
       \path (q3) edge [loop above] node {$$} (q3);
            \path (q4) edge [loop above] node {$ $} (q4);
    \end{tikzpicture}
    \caption{Product of the \emph{Help-me?} game with parameterized automaton with a single memory state realizing $(s_k)_{k \in \N}$
    }
    \label{fig:bestanimal_product}
\end{figure}
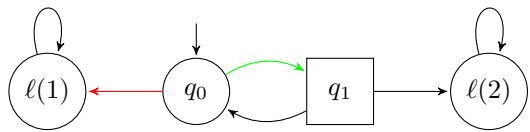

The following theorem shows us that uniform chains indeed suffice to realize any rational behaviour in the sense of maximal chains:

\begin{theorem}
\label{theo:dominatedbyuniformchain}
In a generalised safety/reachability game, every dominated finite memory strategy is dominated by an admissible finite memory strategy or by a maximal uniform chain.
\end{theorem}
\begin{proof}
By Lemma \ref{lemma:belowpreadmissible} it suffices to prove the claim for preadmissible strategies (Definition \ref{def:preadmissible}). We thus start with a preadmissible finite memory strategy $\sigma$.

By the prefix-independence of generalised safety/reachability games, for any combination of vertex $v$ in the game and state $s$ in the automaton realizing $\sigma$, either a history ending in $v$ and state $s$ is a witness for non-admissibility of $\sigma$ or not. Let $N$ be the set of such pairs corresponding to non-admissibility witnesses. By the definition of preadmissibility, we cannot reach any $(v,s) \in N$ without first passing through some $(v,s_v^i)$ with $\aVal(v,s_v^i) = \aVal(v)$ and $\cVal(v,s_v^i) = acVal(v,s_v^i)$. By expanding the automaton if necessary (to remember where we were when first encountering some vertex), we can assume that for any $(v,s) \in N$ there is canonic choice of prior $(v,s_v^i)$.

We now construct either a parameterized automaton from $\sigma$ that either realizes a single maximal strategy, or a maximal uniform chain. If $N$ is empty, we are done. Otherwise, consider $(v,s) \in N$ and the corresponding $(v,s_v^i)$, and compare the associated values: Since the antagonist can reach $(v,s)$ from $(v,s_v^i)$, it has to hold that $\aVal(v,s_v^i) \leq \aVal(v,s) \leq \cVal(v,s) \leq \cVal(v,s_v^i)$. By choice of $(v,s_v^i)$, we have $\aVal(v,s) \leq \aVal(v) = \aVal(v,s_v^i)$, and thus $\aVal(v,s_v^i) = \aVal(v,s)$. Since $(v,s) \in N$, we see that even  $\aVal(v,s_v^i) = \aVal(v,s) = \cVal(v,s) < \cVal(v,s_v^i)$ holds by Lemma \ref{lem:strict_dominance}.

If $\aVal(v,s_v^i) \leq 0$, we modify the automaton to act in $(v,s)$ as it does in $(v,s_v^i)$. If $\aVal(v,s_v^i)$, then we add green edges to let the automaton act in $(v,s)$ as in $(v,s_v^i)$, and red edges to act as it would do originally. The comparison of the values lets us conclude via Lemma \ref{lem:dominance_histories} that the parameterized automaton $\mathcal{M}$ either realizes a single strategy dominated $\sigma$, or a uniform chain dominating $\sigma$.

It remains to argue that the strategy/uniform chain realized by $\mathcal{M}$ is maximal. Let $\sigma_n$ be the strategy where $\mathcal{M}$ is initialized with $n \in \N$ (wlog assume that $n$ is larger than the size of $\mathcal{M}$). Assume that $\tau \succ \sigma_n$, and let $h$ be a witness of $\tau \npreceq \sigma_n$ according to Lemma \ref{lem:dominance_histories}, i.e.~satisfying $\cVal(h,\tau) > \aVal(h,\sigma_n)$. Since $\sigma_n \preceq \tau$, we have $\cVal(h,\sigma_n) \leq \aVal(h,\tau)$, so $\aVal(h,\sigma_n) \leq \cVal(h,\sigma_n) \leq \aVal(h,\tau) \leq \cVal(h,\tau)$ with one inequality being strict. In particular, $h$ is a witness of non-admissibility of $\sigma_n$. By construction of $\mathcal{M}$, $h$ must already have been a witness of non-admissibility of $\sigma$, and the next move after $h$ must be given by a red edge. This already implies that if $\mathcal{M}$ realizes a single strategy, then that strategy is maximal.

Let $m$ be the size of the parameterized automaton $\mathcal{M}$, let $t$ be the size of the automaton realizing $\tau$, and $N = mt + 1$. We compare $\sigma_N$ and $\tau$ by considering the maximal histories compatible with both. If there were such a history $hv$ compatible with both $\sigma_N$ and $\tau$ where $\sigma_N$ is about to apply a red edge, then it has to hold that on histories extending $hv$, $\tau$ always acts at $v$ as $\mathcal{M}$ does following a green edge, for $\tau$ cannot count up to $N$ (in particular, $h$ is maximal for being compatible with $\tau$ and $\sigma_N$). It follows that $\aVal(hv,\tau) \leq 0$. Let $h'v$ be a prefix of this form of $hv$ compatible with $\sigma_n$ not ending in a red edge (this exists, since $n > m$). Then $\aVal(h'v,\tau) \leq 0$, and since $\tau \succeq \sigma_n$, $\aVal(h'v,\sigma_n) = \aVal(v,s_v^i) \leq 0$. But then when constructing $\mathcal{M}$, we would not have placed red and green edges at $(v,s_v^i)$, leading to a contradiction. Thus, at any maximal history compatible with $\sigma_N$ and $\tau$, $\sigma_N$ will follow a green or black edge next.

If $\tau$ is part of a chain $(\tau_i)_{i \in \mathbb{N}}$ with $(\sigma_i)_{i \in \mathbb{N}} \sqsubseteq (\tau_i)_{i \in \mathbb{N}}$, then $\tau$ and $\sigma_N$ have a common upper bound $\tau'$. If some maximal history $h$ compatible with both $\sigma_N$ and $\tau$ is not compatible with $\tau'$, then $h$ has a longest prefix $h'$ compatible with $\tau'$. If $h$ is compatible with $\tau'$, but $\tau'(h) \neq \sigma_N(h)$, we set $h' = h$. As shown above, $h'$ cannot be a witness of non-admissibility of $\sigma_N$, and by Lemma \ref{lem:dominance_histories} it cannot be a witness of non-dominance of $\sigma_N$ by $\tau'$, since $\sigma_N \preceq \tau'$. Lemma \ref{lemma:nonadmissibilitynondominance} then gives us that $h'$ is not a witness of non-dominance of $\tau'$ by $\sigma_N$, i.e.~$\cVal(h',\tau') \leq \aVal(h',\sigma_N)$. Together with $\sigma_N \preceq \tau'$ we get that $\aVal(h',\sigma_N) = \cVal(h',\sigma_N)$. Since $h$ is compatible with $\sigma_N$ and extends $h'$, it follows that $\aVal(h',\sigma_N) = \aVal(h,\sigma_N) = \cVal(h,\sigma_N)$. Since $\tau \preceq \tau'$, it follows that $\cVal(h',\tau) \leq \cVal(h',\tau') = \aVal(h',\sigma_N)$. Since $h$ is compatible with $\tau$ and extends $h'$, it follows that $\cVal(h,\tau) \leq \cVal(h',\tau) \leq \aVal(h',\sigma_N) = \aVal(h,\sigma_N)$, i.e.~that $h$ is not a witness of non-dominance of $\tau$ by $\sigma_N$.

The remaining case we need to consider is that where $h$ is compatible with $\tau'$ and $\tau'(h) = \sigma_N(h)$. Consider the subgame starting after that move. Since we have chosen $N$ sufficiently big, in this subgame it is impossible for $\sigma_N$ to pass through a red edge without passing through a green edge at the same vertex. By construction, this ensures that $\sigma_N$ is still preadmissible in this subgame. Since reaching the subgame is compatible with $\tau'$ and $\sigma_N$, restricting to this subgame, we still have that $\sigma_N \preceq \tau'$. Thus, we can apply Lemma \ref{lemma:valuescoincide} to the subgame, and conclude that $\aVal(h,\tau') = \aVal(h,\sigma_N)$ and $\cVal(h,\tau') = \cVal(h,\sigma_N)$. Since $h$ cannot be a witness of non-dominance of $\tau$ by $\tau'$, it holds that $\cVal(h,\tau) \leq \aVal(h,\tau') = \aVal(h,\sigma_N)$. Thus, $h$ is not a witness of non-dominance of $\tau$ by $\sigma_N$ either.

As we have ruled out all candidates for witnesses of non-dominance of $\tau$ by $\sigma_N$, by Lemma \ref{lem:dominance_histories} we may conclude that $\tau \preceq \sigma_N$.
\end{proof}

Theorem~\ref{theo:dominatedbyuniformchain} cannot be extended to state that every chain comprised of finite memory strategies is below an admissible strategy or a maximal uniform chain. Note that there are only countably many uniform chains.

\begin{example}
There is a generalised safety/reachability game where there are uncountably many incomparable maximal chains of finite memory strategies.
\begin{proof}
Consider the game depicted in Figure \ref{fig:bestanimal2choices_1}. For any $p \in \{a,b\}^\omega$, define a chain of finite memory strategies by letting the $n$-strategy be \emph{loop $n$ times while playing the symbols from $p_{\leq n}$, then quit}. For each $p$, we obtain a different maximal chain.
\end{proof}
\end{example}

The particular structure of parameterized automata over safety/reachability game graphs lets us point out useful properties and patterns, that we present before tackling algorithmic properties in the next section.
Notice that for each path $\tilde{h}$ in $\mathcal{M} \times G$, there exists a unique corresponding path $h$ in $G$, as there may exist a transition between $(m,v)$ and $(m',v')$ in $\mathcal{M} \times G$ only if $(v,v') \in E$, by definition of the product.

\begin{definition}\label{def:valid_path}
Let  $\game$ be a generalised safety/reachability game and $\mathcal{M}$ be a parameterized automaton over the game graph $G$.
Then, we say that a (finite or infinite) path $\tilde{h}$ in the product $\mathcal{M} \times G$ is \emph{valid} if it respects the following two conditions:
\begin{enumerate}
\item there is a finite number of green transitions in $\tilde{h}$, and
\item no green transition appears after a red transition in $\tilde{h}$.
\end{enumerate}
\end{definition}

If a path $\tilde{h}$ is valid then there exists a strategy in the sequence realized by $\mathcal{M}$ that is compatible with the corresponding path $h$ in $G$.
Furthermore, if $\tilde{h}$ does not contain any red transition, then $h$ is compatible with an infinite number of histories in the sequence, namely all strategies corresponding to the counter initialized to a value greater than the number of green transitions in $\tilde h$.
For a valid path $\tilde h$, let $| \tilde h|_g$ denote the number of green transitions in $\tilde h$, and $| \tilde h|_r$ the number of red transitions in $\tilde h$.
Furthermore, a valid path $\tilde h$ can be decomposed in two valid paths, that is, $\tilde h = \tilde h_g \tilde h_r$ where $|\tilde h_g|_r=0$ and $| \tilde h_r|_g =0$, such that $(\tilde h_r)_0=(\mu(\last{\tilde h_g},0),\nu(\last{\tilde h_g},0))$.
If $| \tilde h|_r=0$, we have $\tilde h = \tilde h_g$ and $\tilde h_r = \varepsilon$.

\begin{lemma}\label{lem:compatibility}
Let  $\game$ be a generalised safety/reachability game and $\mathcal{M}$ be a parameterized automaton over the game graph $G$.
Let $\tilde h$ be a valid path in $\mathcal{M} \times G$ such that $\tilde h_0 = (m_0,v_0)$ and let $k \in \N$.
Then, the history $h$ is compatible with $\sigma_k$ if, and only if, $| \tilde h |_r =0$ and $k\geq | \tilde h |_g$, or $| \tilde h |_r > 0$ and $k =   | \tilde h |_g $.
\end{lemma}
\begin{proof}
Let $\tilde h$ be a valid path in $\mathcal{M} \times G$ and let $k \in \N$.
\begin{itemize}

\item[$\Rightarrow$]
Assume that the history $h$ is compatible with $\sigma_k$.
\begin{enumerate}
\item If $| \tilde h |_r =0$, and $k < | \tilde h |_g$, then there exists $n < |\tilde h|$ such that $| \tilde h_{\leq n}|_g =k$ and $\sigma_k( h_{\leq n})=\nu(\tilde h_{n},0)$.
Since $\tilde{h}$ is a valid path, by condition $2.$ in Definition~\ref{def:valid_path}, we have that $\tilde h_{n+1} = (\mu(\tilde h_n,1), \nu(\tilde h_n,1))$.
Hence $h_{n+1} = \nu(\tilde{h}_n,1)$ and we have that $h_{n+1} = \nu(\tilde{h}_n,1) \neq \nu(\tilde h_{n},0) = \sigma_k(\tilde h_{\leq n})$.
That is, $h$ is not compatible with $\sigma_k$ which is a contradiction.
\item If $| \tilde h |_r > 0$, and $k \neq   | \tilde h |_g $, we have two cases:
If $k <   | \tilde h |_g $, there exists a smallest $n < |\tilde h|$ such that $| \tilde h_{\leq n}|_g =k$,  and $\sigma_k(\tilde h_{\leq n})=\nu(\tilde h_{n},0)$.
Since $\tilde{h}$ is a valid path and $k <   | \tilde h |_g $, by condition $2.$ in Definition~\ref{def:valid_path}, we have that $\tilde h_{n+1} = (\mu(\tilde h_n,1), \nu(\tilde h_n,1))$.
Hence $h_{n+1} = \nu(\tilde{h}_n,1)$ and we have that $h_{n+1} = \nu(\tilde{h}_n,1) \neq \nu(\tilde h_{n},0) = \sigma_k( h_{\leq n})$.
That is, $h$ is not compatible with $\sigma_k$ which is a contradiction.
If $k >   | \tilde h |_g $, there exists a smallest $n < |\tilde h|$ such that $\tilde h_{n+1} = (\mu(\tilde h_n,0), \nu(\tilde h_n,0))$.
Since  $k >   | \tilde h |_g $, we have that $\sigma_k(h_{\leq n})=\nu(\tilde{h}_{\leq n},1)$.
Hence  we have that $h_{n+1} = \nu(\tilde{h}_n,0) \neq \nu(\tilde h_{n},1) = \sigma_k( h_{\leq n})$.
That is, $h$ is not compatible with $\sigma_k$ which is a contradiction.

\end{enumerate}

\item[$\Leftarrow$]

\begin{enumerate}
\item Assume $| \tilde h |_r =0$ and $k\geq | \tilde h |_g$.
Then for all $n < |\tilde h|$ such that $\tilde{h}_n \in M \times V_i$, we have that $\tilde h_{n+1} = (\mu(\tilde{h}_n, 1),\nu(\tilde{h}_n, 1))$.
Since $k \geq  | \tilde h |_g$, we also have that $\sigma_k(h_n)=\nu(\tilde{h}_n, 1) = h_{n+1}$ (that is, the counter value cannot have been decremented $k$ times by following $\tilde{h}$) .
As $\tilde h_0 =(m_0,v_0)$, it follows that $\tilde h$ is indeed compatible with $\sigma_k$.

\item Assume $| \tilde h |_r > 0$ and $k =   | \tilde h |_g $.
Recall that $\tilde{h}$ can be decomposed in $\tilde h_g \tilde h_r$ where $|\tilde h_g|_r=0$ and $| \tilde h_r|_g =0$, such that $(\tilde h_r)_0=(\mu(\last{\tilde h_g},0),\nu(\last{\tilde h_g},0))$.
As  $|\tilde h_g|_g=|\tilde h|_g= k$ by case $1.$ above, we know that $h_g$ is compatible with $\sigma_k$.
Furthermore, the counter value at $\last{h_g}$ is $0$ when initialized at $k$.
Thus, for all  $|\tilde h_g| \leq n < |\tilde h|$ such that $\tilde{h}_n \in M \times V_i$, we have that $\tilde h_{n+1} = (\mu(\tilde{h}_n, 0),\nu(\tilde{h}_n, 0))$.
Consequently, we also have that $\sigma_k(h_n)=\nu(\tilde{h}_n, 0) = h_{n+1}$ (that is, the counter value has reached $0$ after $\tilde h_g$ and then follows $\tilde{h}_r$).
As $\tilde h_0 =(m_0,v_0)$, it follows that $\tilde h$ is indeed compatible with $\sigma_k$.

\end{enumerate}

\end{itemize}
\end{proof}

\subsection{Algorithmic properties}
\label{subsec:algorithms}

In this section, we prove two decidability results concerning parametrized automata.

First, we prove that we can decide whether the sequence of strategies realized by a parameterized automaton is a chain.
Note that this decision problem is not trivial: not every parameterized automaton realizes an (increasing) chain of strategies.
For instance, if we switch the red and green transitions in the automaton/game graph product of figure~\ref{fig:bestanimal_product},
the sequence of strategies realized consists of $s_\omega$ when the counter is initialized with value $0$,
and $s_0$ when it is initialized with any other value.
As $s_\omega \not \preceq s_0$, it is not a chain.

Second, we demonstrate that we can compare uniform chains: given two parametrized automata defining chains of strategies, we can decide whether one is dominated by the other.
We begin by proving that strategies realized by Mealy automata are comparable.

\begin{lemma}\label{lem:comparison_two_strategies}
Let $\game$ be a generalised safety/reachability game,
let $\sigma$ and $\sigma'$ be finite-memory strategies realized by the finite Mealy automata $\mathcal{M}$ and $\mathcal{M}'$.
It is decidable in $\textsc{PTime}$ whether $\sigma \preceq \sigma'$.
\end{lemma}
\begin{proof}
We construct the game $\game'$ of perfect information for two players, \emph{Challenger} and \emph{Prover}, such that Prover wins the game if and only if $\sigma \preceq \sigma'$.
Let $\mu^\mathcal{M}$ and $\nu^\mathcal{M}$ be the memory update and move choice functions of $\mathcal{M}$, and let $\mu^{\mathcal{M}'}$ and $\nu^{\mathcal{M}'}$ be the memory update and move choice functions of $\mathcal{M}'$.
Formally, the game $\game'=\langle P', G', W_C \rangle $ is such that the set of players $P'$ is composed of Challenger and Prover, and the game graph $G' = (V',E')$ is composed of:
\begin{itemize}
\item a copy of the graph $(\mathcal{M} \times G \times \mathcal{M}')$,
\item for each $c \in L \cup \lbrace 0 \rbrace$, a copy $(\mathcal{M} \times G)^c$ of the graph $(\mathcal{M} \times G)$,
\item for each $a \in L \cup \lbrace 0 \rbrace$, a copy $(\mathcal{M}' \times G)^a$ of the graph $(\mathcal{M}' \times G)$,
\item for each pair $(c,a) \in L \cup \lbrace 0 \rbrace \times L \cup \lbrace 0 \rbrace$ such that $c > a$, and for each vertex $(m^\mathcal{M},v, m^{\mathcal{M}'})$ of the graph $(\mathcal{M} \times G \times \mathcal{M}')$ such that $Succ(m^\mathcal{M},v, m^{\mathcal{M}'}) = \emptyset$:
\begin{itemize}
\item a vertex $(c,a,m^\mathcal{M},v, m^{\mathcal{M}'})$,
\item an edge $((m^\mathcal{M},v, m^{\mathcal{M}'}),(c,a,m^\mathcal{M},v, m^{\mathcal{M}'}))$,
\item an edge $((c,a,m^\mathcal{M},v, m^{\mathcal{M}'}),(m'^{\mathcal{M}},v')^c)$ where $(m'^{\mathcal{M}},v') = (\mu^\mathcal{M} (m^\mathcal{M},v), \nu^\mathcal{M}(m^\mathcal{M},v))$ (ie, $(m'^{\mathcal{M}},v')$ is the successor of $(m^\mathcal{M},v)$ in $(\mathcal{M} \times G)$),
\item an edge $((c,a,m^\mathcal{M},v, m^{\mathcal{M}'}), (m'^{\mathcal{M}'},v')^a)$, where $(m'^{\mathcal{M}'},v') = (\mu^{\mathcal{M}'} (m^{\mathcal{M}'},v), \nu^{\mathcal{M}'}(m^{\mathcal{M}'},v))$ (ie, $(m'^{\mathcal{M}'},v')$ is the successor of $(m^{\mathcal{M}'},v)$ in $(\mathcal{M}' \times G)$).
\end{itemize}
\end{itemize}
All vertices in $V'$ belong to Challenger except the ones of the form $(c,a,m^\mathcal{M},v, m^{\mathcal{M}'})$, that belong to Prover.
As the sets $L \cup \lbrace 0 \rbrace$ is finite, there are also a finitely many pairs $(c,a) \in L \cup \lbrace 0 \rbrace \times L \cup \lbrace 0 \rbrace$ such that $c>a$.
Thus, we see that we wind up with a game graph $G'$ polynomial in the size of the original strategy automata $\mathcal{M}$ and $\mathcal{M}'$ and game graph $G$.

The winning condition $W_C$ for Challenger can be expressed as a boolean combination of safety and reachability conditions.
First, Challenger must reach a vertex of the form $(c,a,m^\mathcal{M},v, m^{\mathcal{M}'})$.
Then, if both $c,a \neq 0$ he must reach either any vertex of the form $(m^\mathcal{M}, \ell(c))^c$ in $(\mathcal{M} \times G)^c$, or reach any vertex of the form $(m^{\mathcal{M}'}, \ell(a))^a$ in $(\mathcal{M}' \times G)^a$.
If $c =0$ and $a \neq 0$, he must either:
\begin{itemize}
\item reach any vertex in $(\mathcal{M} \times G)^c$ (to ensure he is playing in the subgame $\mathcal{C}$) \emph{and} be safe from all vertices of the form $(m^\mathcal{M}, \ell(x))^c$  with $x \neq 0$ in $(\mathcal{M} \times G)^c$,
\item or reach any vertex in $(\mathcal{M}' \times G)^a$ (to ensure he is playing in the subgame $\mathcal{A}$) \emph{and} reach any vertex of the form $(m^{\mathcal{M}'}, \ell(a))^a$ in $(\mathcal{M}' \times G)^a$.
\end{itemize}
The case $c \neq 0, a =0$ can be expressed in similar terms.
%The game $\game'$ is in fact a zero-sum boolean game for two players with an $\omega$-regular objective, thus the decision problem for $\game'$ is decidable.
Solving the game $\game'$ amounts to solve, for each reachable vertex of the form $(c,a,m^\mathcal{M}, v, {m}^{\mathcal{M}'})$, one reachability (or safety in the case $c=0$) game in the subgame $(\mathcal{M} \times G)^c$ initialized in $(\mu^\mathcal{M} (m^\mathcal{M},v), \nu^\mathcal{M}(m^\mathcal{M},v))^c$ and one reachability (or safety in the case $a=0$) game in the subgame $(\mu^{\mathcal{M}'} (m^{\mathcal{M}'},v), \nu^{\mathcal{M}'}(m^{\mathcal{M}'},v))^a$ initialized in $(m^{\mathcal{M}'},v)$.
Thus, $\game'$ is solvable in $\textsc{PTime}$.

\medskip
We claim that Challenger wins $\game'$ if, and only if $\sigma \not \preceq \sigma'$.

\begin{itemize}
\item[$\Rightarrow$] Let $S$ be a winning strategy of Challenger in $\game'$.
This means that for each strategy $T$ of Prover, the outcome $\rho$ of $(S,T)$ satisfies the winning condition $W_C$.
Thus, we know that every such $\rho$ leaves $(\mathcal{M} \times G \times \mathcal{M}')$ after a finite number of steps, that is, there exists $k_{\rho} \in \N$ such that $\rho_{k_{\rho}}$ is of the form $(c,a,m^\mathcal{M},v, m^{\mathcal{M}'})$.
As Prover has only one move to make in each outcome, which occurs after the choice of $c$ and $a$ by Challenger, there actually are only two different possible outcomes for $S$: the outcome $\rho_c$, which corresponds to Prover choosing to contest $c$, and the outcome $\rho_a$, which corresponds to Prover choosing to contest $a$.
The two outcomes $\rho_c$ and $\rho_a$ share a longest common prefix $ h' \in V'^*$ that satisfies that $h' = \tilde {h} (c,a,m^\mathcal{M},v, m^{\mathcal{M}'})$ where $\tilde {h} \in (\mathcal{M} \times G \times \mathcal{M}')$, $|\tilde{h}|=k=k_{\rho_c} = k_{\rho_a}$, $\last {\tilde {h}} = (m^\mathcal{M},v, m^{\mathcal{M}'})$ and $Succ(m^\mathcal{M},v, m^{\mathcal{M}'}) = \emptyset$
%$m\in M_C$
(otherwise, there is no transition to the vertex $(c,a,m^\mathcal{M},v, m^{\mathcal{M}'})$ and no way to reach the subgames $\mathcal{C}$ or $\mathcal{A}$).

Furthermore, as $S$ is winning, there exist two infinite paths $\tilde {h}^c$ in $(\mathcal{M} \times G)^c$ and $\tilde {h}^a$ in $(\mathcal{M}' \times G)^a$ such that $\rho_c= \tilde {h} (c,a,m^\mathcal{M},v, m^{\mathcal{M}'}) \tilde{ h}^c$ and $\rho_a= \tilde {h}(c,a,m^\mathcal{M},v, m^{\mathcal{M}'}) \tilde{h}^a$.
Let $h$ be the history in $G$ corresponding to $\tilde {h}$ in $(\mathcal{M} \times G \times \mathcal{M}')$ and let $h^c$ and $h^a$ be the infinite paths in $G$ corresponding to $\tilde{ h}^c$ and $\tilde{ h}^a$.
Consider now the outcomes $h h^c$ and $h h^a$ in $\game$.
By construction, the outcome $h h^c$ is compatible with $\sigma$, and yields payoff $c$. %(In fact, if $c \neq 0$, we even know that $c' =c$.)
Thus, $\cVal(h, \sigma)\geq c$.
By construction, the outcome $h h^a$ is compatible with $\sigma'$ and yields payoff $a$. %(In fact, if $a \neq 0$, we even know that $a' =a$.)
Thus, $\aVal(h, \sigma') \leq a$.
Finally, we have that $h$ is compatible with $\sigma$ and $\sigma'$, $\sigma(h) \neq \sigma'(h)$ and $\cVal(h, \sigma) > \aVal(h, \sigma') $, that is, $h$ is a non-dominance witness for $\sigma$ and $\sigma'$.
Thus, $\sigma \not \preceq \sigma'$.

\item[$\Leftarrow$] Suppose $\sigma \not \preceq \sigma'$.
There exists a history $h$ compatible with $\sigma$ and $\sigma'$ such that  $\sigma(h) \neq \sigma'(h)$ and $\cVal(h,\sigma) > \aVal(h,\sigma')$.
Let $c = \cVal( h,\sigma)$ and $a = \aVal( h,\sigma')$.
There exists an infinite path $h^c$ from $\last {h}$ such that the outcome $hh^c$ is compatible with $\sigma$ and yields payoff $c$.
Hence, there exists a corresponding infinite path $\tilde{h}^c$ from $\last {\tilde{h}}$ in $(\mathcal{M} \times G)$, and thus a corresponding path in $(\mathcal{M} \times G)^c$.
Similarly, there exists an infinite path $h^a$ from $\last {h}$ such that the outcome $hh^a$ is compatible with $\sigma'$ and yields payoff $a$. Hence, there exists a corresponding infinite path $\tilde{h}^a$ from $\last {\tilde{h}}$ in $(\mathcal{M}' \times G)$ and thus a corresponding path in $(\mathcal{M}' \times G)^a$.
Let $\sigma'$ be the strategy of Challenger in $\game'$ that:
\begin{itemize}
\item in the first phase, follows $\tilde{h}$ in  $(\mathcal{M} \times G \times \mathcal{M}')$,
\item at $\last {\tilde{h}}= (m^\mathcal{M},v, m^{\mathcal{M}'})$, chooses the successor vertex to be $(c,a,m^\mathcal{M},v, m^{\mathcal{M}'})$,
\item in $(\mathcal{M} \times G)^c$, follows $\tilde{h}^c$,
\item and in $(\mathcal{M}' \times G)^a$, follows $\tilde{h}^a$.
\end{itemize}
The strategy $\sigma'$ is winning for Challenger.
Indeed, since the values $c$ and $a$ are accurate, the existence of both $\tilde{h}^c$ and $\tilde{h}^a$ are guaranteed.
Furthermore, if $c\neq 0$, a vertex of the form $(m^\mathcal{M}, \ell(c))^c$ appears in $\tilde{h}^c$, thus the outcome $\rho_c$ of $\game'$ is winning.
If $c=0$, either a vertex of the form $(m^\mathcal{M}, \ell(0))^c$ appears in $\tilde{h}^c$ or $\tilde{h}^c$ avoids all leaves,  thus the outcome $\rho_c$ of $\game'$ is winning.
Similarly, we can see that $\rho_a$ is also winning, regardless of $a$ being equal to $0$ or not.
Thus the outcome of $\game'$ with Challenger playing according to $S$ is winning regardless of the strategy of Prover, that is, his choice between contesting $c$ or contesting $a$.

\end{itemize}
\end{proof}

We now expose equivalences between the decision problems we are interested in,
and properties $(\textsf{P}_1)$, $(\textsf{P}_2)$, $(\textsf{P}_3)$ and $(\textsf{P}_4)$
that can be decided with the use of Lemma~\ref{lem:comparison_two_strategies}.

\begin{proposition}\label{prop:bound}
Let $\game$ be a generalised safety/reachability game over a graph $G$.
Let $\mathcal{M}$ be a Mealy automaton realizing a finite memory strategy $M$,
and let $\mathcal{S}$ and $\mathcal{T}$ be parameterized automata realizing sequences $(S_n)_{n\in\N}$ and $(T_n)_{n\in\N}$ of finite memory strategies.
Then:
\begin{enumerate}[topsep=0pt]
\item\label{prop:bound_chain}
Let $N_{\preceq} = |G||\mathcal{S}|$.

Then $(S_n)_{n\in\N}$ is a chain if and only if $(\textsf{P}_1)$ $T_{i} \preceq T_{i+1}$ for every $1 \leq i \leq N_{\preceq}$.
\item\label{prop:bound_inc_chain}
Let $N_{\prec} = |G||\mathcal{S}| + (|G||\mathcal{S}|)!$.

Then $(S_n)_{n\in\N}$ is an increasing chain if and only if $(\textsf{P}_2)$ $T_{i} \prec T_{i+1}$ for every $1 \leq i \leq N_{\prec}$.
\item\label{prop:bound_comp1}
Let $N_{T} = |G||\mathcal{T}|(|\mathcal{M}| + 1) + 1$, and suppose that $(T_n)_{n \in \N}$ is a chain.

Then $M \not \sqsubseteq (T_n)_{n \in \N}$ if and only if $(\textsf{P}_3)$ $M \not \preceq T_{N_{T}}$.
\item\label{prop:bound_comp2}
Let $N_{\mathcal{S}} = |G||\mathcal{S}|(2|\mathcal{T}|+1)$, and suppose that $(S_n)_{n \in \N}$ and $(T_n)_{n \in \N}$ are chains.

Then $(S_n)_{n \in \N} \not \sqsubseteq (T_n)_{n \in \N}$ if and only if $(\textsf{P}_4)$ $S_{N_{\mathcal{S}}} \not \preceq (T_n)_{n \in \N}$.
\end{enumerate}
\end{proposition}

The proof of Proposition~\ref{prop:bound}.\ref{prop:bound_chain} and ~\ref{prop:bound}.\ref{prop:bound_inc_chain}
is based on the following auxiliary Lemma, whose demonstration relies on the study of the loops
that appear in witnesses of non dominance.

\begin{lemma}\label{lem:iteration}
Let $\game$ be a generalised safety/reachability game, let $\mathcal{M}$ be a parametrized automaton over the game graph of $\game$,
and let $(T_n)_{n \in \mathbb{N}}$ be the sequence of finite-memory strategies realized by $\mathcal{M}$.
Then for every pair of integers $n_1,n_2 > |G||\mathcal{M}|$ satisfying $T_{n_1} \not \preceq T_{n_2}$,
there exists $0 < k \leq |G||\mathcal{M}|$ such that for every $i \in \mathbb{N}$, $T_{n_1 + (i-1)k} \not \preceq T_{n_2 + (i-1)k}$.
\end{lemma}

\begin{proof}
Let $n_1,n_2 > |G||\mathcal{M}|$, and let us suppose that $T_{n_1} \not \preceq T_{n_2}$.
By Lemma~\ref{lem:dominance_histories}, we know that there exists a non-dominance witness $h$ of $T_{n_1}$ by $T_{n_2}$,
i.e., a history $h$ compatible with $T_{n_1}$ and $T_{n_2}$ such that $T_{n_1}(h) \neq T_{n_2}(h)$ and $\cVal(h, T_{n_1}) > \aVal(h, T_{n_2})$.
We expose an integer $k \leq |G||\mathcal{M}|$ such that for every $i \in \mathbb{N}$, we are able to construct, based on $h$, a non-dominance
witness of $T_{n_1 + (i-1)k}$ by $T_{n_2 + (i-1)k}$.

Let $\tilde{h}_1$, respectively $\tilde{h}_2$, be the valid path of $G \times \mathcal{M}$ corresponding to the history $h$
and the initial counter values $n_1$, respectively $n_2$.
By supposition, both $n_1$ and $n_2$ are strictly greater than $|G||\mathcal{M}|$.
As a consequence, since $T_{n_1}$ and $T_{n_2}$ have an identical behaviour as long as both have not emptied their counter,
and they disagree on the history $h$, there exists a common prefix $\tilde{h}'$ of $\tilde{h}_1$ and $\tilde{h}_2$ such that $|\tilde{h}'|_{g} = |G||\mathcal{M}| + 1$
and  $|\tilde{h}'|_{r} = 0$.
Since $|\tilde{h}'|_{g} = |G||\mathcal{M}|$, there exists a state $(v,m) \in G \times \mathcal{M}$
that appears at least twice in $\tilde{h}'$ just after using a green transition.
This repetition yields a decomposition $\tilde{h}_{a}\tilde{h}_{b}\tilde{h}_{c}$ of $\tilde{h}'$
such that $\last{\tilde{h}_a} = \last{\tilde{h}_b} = (v,m)$, and $0 < |\tilde{h}_b|_{g} \leq |G||\mathcal{M}|$.
Let $k$ denote $|\tilde{h}_{b}|_{g}$.

For every $i \in \mathbb{N}$, let $h(i)$ denote the history $h_a(h_b)^{i}h_c$ and let $\tilde{h}(i)$ denote the valid path $\tilde{h}_a(\tilde{h}_b)^{i}\tilde{h}_c$.
For any valid continuation $\tilde{h}''$ of $\tilde{h}'$, the path $\tilde{h}(i)\tilde{h}''$ is valid
and $|\tilde{h}(i)\tilde{h}''|_{g} = |\tilde{h}'\tilde{h}''|_{g} + (i-1)k$.
As a consequence, by Lemma \ref{lem:compatibility}, for every history $h''$, for every $n \in \mathbb{N}$,
$h'h''$ is compatible with $T_{n}$ if and only if $h(i)h''$ is compatible with $T_{n + (i-1)k}$.
More generally, for any outcome $\rho$, the outcome $h'\rho$ is compatible with $T_{n}$
if and only if $h(i)\rho$ is compatible with $T_{n + (i-1)k}$, and since $\game$ is prefix-independent, the payoff of these outcomes are equal.
As a consequence, $\cVal(h(i)h'', T_{n + (i-1)k}) = \cVal(h, T_{n})$ and $\aVal(h(i)h'', T_{n + (i-1)k}) = \aVal(h, T_{n})$.
In particular, if we let $h''$ be the history such that $h_ah_bh_ch'' = h$, we obtain
\begin{itemize}
\item
$h(i)h''$ is compatible with $T_{n_1 + (i-1)k}$ and $T_{n_2 + (i-1)k}$;
\item
$T_{n_1 + (i-1)k}(h(i)h'') = T_{n_1}(h) \neq T_{n_2}(h) = T_{n_2 + (i-1)k}(h(i)h'')$;
\item
$\cVal(h(i)h'', T_{n_1 + (i-1)k}) = \cVal(h, T_{n_1}) > \aVal(h, T_{n_2}) = \aVal(h(i)h'', T_{n_2 + (i-1)k})$.
\end{itemize}

Therefore $h(i)h''$ is a witness of non-dominance of $T_{n_{1} + (i-1)k}$ by $T_{n_{2} + (i-1)k}$, hence $T_{n_1 + (i-1)k} \not \preceq T_{n_2 + (i-1)k}$.
\end{proof}

\begin{proof}[Proof of Proposition~\ref{prop:bound}.\ref{prop:bound_chain}
 and \ref{prop:bound}.\ref{prop:bound_inc_chain}]
Let $\game$ be a generalised safety/reachability game, and let $\mathcal{S}$ be a parametrized automaton over the game graph of $\game$.
We denote by $(S_n)_{n \in \mathbb{N}}$ the sequence of finite-memory strategies realized by $\mathcal{S}$.
Let $N_{\preceq} = |G||\mathcal{S}|$.
 and $N_{\prec} = |G||\mathcal{S}| + (|G||\mathcal{S}|)!$.

Let $U_{\mathcal{S}}$ denote the set composed of the integers $n$ satisfying $S_{n} \not \preceq S_{n+1}$.
It is clear that if $U_{\mathcal{S}}$ is not empty, then $(S_n)_{n \in \mathbb{N}}$ is not a chain.
Conversely, if $U_{\mathcal{S}}$ is empty, then $(S_n)_{n \in \mathbb{N}}$ is a chain, since for every pair of integers
$n_1 < n_2$, we have $S_{n_1} \preceq S_{n_1+1} \preceq \ldots \preceq S_{n_2}$.
Let us suppose, towards building a contradiction, that the minimal element $m$ of $U_{\mathcal{S}}$ is strictly greater than $N_{\preceq}$.
Then, we obtain from Lemma \ref{lem:iteration} that there exists an integer $k > 0$ such that $S_{m-k} \not \preceq S_{m-k+1}$ by setting $i = 0$.
This contradicts the minimality of $m$.
As a consequence, $m \leq N_{\preceq}$.
This proves that $(S_n)_{n \in \mathbb{N}}$ is a chain if and only if $S_{i} \preceq S_{i+1}$ for every $1 \leq i \leq N_{\preceq}$.

Let $V_{\mathcal{S}}$ be the set of integers $n$ satisfying $S_{n + 1} \preceq S_{n}$.
As before, we obtain that $(S_n)_{n\in\mathbb{N}}$ is an increasing chain if and only if $U_{\mathcal{S}} \cup V_{\mathcal{S}}$ is empty.
To conclude, let us suppose, towards building a contradiction, that the minimal element $m$ of $V_{\mathcal{S}}$ is strictly greater than $N_{\prec}$.
Let $m'$ denote $m-(|G||\mathcal{S}|)!$.
Then, since $|G||\mathcal{S}| < m' < m$, $S_{m' + 1} \not \preceq S_{m'}$ by minimality of $m$,
and we obtain from Lemma \ref{lem:iteration} that there exists an integer $k \leq |G||\mathcal{S}|$
such that for every $i \in \mathbb{N}$, $S_{m' + (i-1)k + 1} \not \preceq S_{m' + (i-1)k}$.
In particular, since $k$ divides $(|G||\mathcal{S}|)!$ and $m = m' + (|G||\mathcal{S}|)!$, we have $S_{m + 1} \not \preceq S_{m}$, which contradicts the fact that
$m$ is in $V_{\mathcal{S}}$.
As a consequence, $m \leq N_{\prec}$.
This proves that $(S_n)_{n \in \mathbb{N}}$ is an increasing chain if and only if $S_{i} \prec S_{i+1}$ for every $1 \leq i \leq N_{\prec}$.
\end{proof}

\begin{proof}[Proof of Proposition~\ref{prop:bound}.\ref{prop:bound_comp1}]

\begin{itemize}
\item[$\Rightarrow$] trivial.
\item[$\Leftarrow$] Assume that $M \not \preceq T_{N_{T}}$.
By Lemma~\ref{lem:dominance_histories}, we know that there exists a non-dominance witness $h_{N_{T}}$ for $M$
and $T_{N_{T}}$ such that $h_{N_{T}}$ is compatible with $M$ and $T_{N_{T}}$, $M(h_{N_{T}}) \neq T_{N_{T}}(h_{N_{T}})$ and $\cVal(h_{N_{T}}, M) > \aVal(h_{N_{T}}, T_{N_{T}})$.
We claim that from $h_{N_{T}}$ we are able to construct non-dominance witnesses for an infinite number of strategies in the chain $(T_n)_{n\in\mathbb{N}}$.

Let $\rho$ be an outcome of $\game$ such that $h_{N_{T}} \rho$ is compatible with $T_{N_{T}}$ and yields the antagonistic value $\aVal(h_{N_{T}},T_{N_{T}})$.
We consider the valid path $\tilde{h}_{N_{T}}$ in $G \times \mathcal{T}$ that follows the behaviour of $T_{N_{T}}$ over the history $h$,
the valid path $\hat{h}_{N_{T}}$ in $\mathcal{M} \times G \times \mathcal{T}$ that follows the behaviour of $M$ and $T_{N_{T}}$ over the history $h$,
and the valid path $\tilde{\rho}$ in $G \times \mathcal{T}$ that follows the behaviour of $T_{N_{T}}$ over the outcome $\rho$.
We have one of the following:
either $(a)$ $|\tilde h_{N_{T}}\tilde{\rho}|_g = 0$  or $(b)$ $|\tilde h_{N_{T}}\tilde{\rho}|_r \geq 1$.
We show now how to construct new non-dominance witnesses:
\begin{itemize}
\item[$(a)$] Assume that $|\tilde h_{N_{T}}\tilde{\rho}|_g < {N_{T}}$.
Then we know that $h_{N_{T}}$ is also a non-dominance witness for $M$ and $T_n$ for all $n \geq N'$ with $N'= | \tilde h_{N_{T}} \tilde \rho|_g$.
Indeed, by Lemma~\ref{lem:compatibility}, the outcome $h_{N_{T}}\rho$ is compatible with every strategy realized by $\mathcal{T}$
initialized with a counter value at least equal to the number of green transitions in $\tilde h_{N_{T}} \tilde\rho$, that is, $N'$.
Furthermore, for $ n \geq N'$, we have that $\aVal(h_{N_{T}}, T_n) \leq p_i(h_{N_{T}}\rho)$ by definition of the antagonistic value.
Thus, $\aVal(h_{N_{T}}, T_n) \leq \aVal(h_{N_{T}},T_{N_{T}})$, and as $\cVal(h_{N_{T}},M) > \aVal(h_{N_{T}},T_{N_{T}})$ we also have $\cVal(h_{N_{T}},M) > \aVal(h_{N_{T}},T_n)$.
Recall finally that since $T_{N_{T}}(h_{N_{T}}) = T_n(h_{N_{T}})$ and $M(h_{N_{T}}) \neq T_n(h_{N_{T}})$, we have $M(h_{N_{T}}) \neq  T_n(h_{N_{T}})$.
Hence, the history $h_{N_{T}}$ is also a non-dominance witness for $M$ and $T_n$.
There are infinitely many $n$ such that $n > N'$, thus we obtain that $M \not \sqsubseteq (T_n)_{n\in\mathbb{N}}$.

\item[$(b)$] Assume now that $|\tilde h_{N_{T}}\tilde{\rho}|_g \geq {N_{T}} = |G| |\mathcal{T}| (|\mathcal{M}|  + 1) + 1$.
Then at least one of the two following properties is verified: either $|\tilde h_{N_{T}}|_g > |\mathcal{M}| |G| |\mathcal{T}|$ and we expose a loop in $\hat h_{N_{T}}$,
or $|\tilde{\rho}|_g > |G| |\mathcal{T}|$ and we expose a loop in $\tilde{\rho}$.
In both cases we construct non-dominance witnesses by iterating the loop.
\begin{enumerate}
\item
Assume that $|\tilde h_{N_{T}}|_g > |\mathcal{M}| |G| |\mathcal{T}|$.
This means that there exists a state $(m^\mathcal{M}, v, m^\mathcal{T})$ that appears at least twice in $\hat h_{N_{T}}$ after following a green transition, hence before any red transition appears in $\hat h_{N_{T}}$.
Thus there exists a decomposition $h_{N_{T}} = h^1 h^{loop} h^2$ of $h_{N_{T}}$ such that $|\hat h^{loop}|_g \geq 1$ and $\last { \hat h^{loop}} = \last { \hat h^1} = (m^\mathcal{M}, v, m^\mathcal{T})$.
(note that possibly $h^2 = \varepsilon$).
Let $k \in \N$ and consider the history $h_k = h^1 (h^{loop})^k h^2$.
Since $h_{N_{T}}$ is compatible with $M$ and $\last { \hat h^{loop}} = \last { \hat h_1} \in \mathcal{M} \times G \times \mathcal{T}$, we know that $h_k$ is compatible with the strategy $M$ for any integer $k$.
Moreover, all continuation paths after $h_{N_{T}}$ for $M$ are also available after $h_{k}$, and conversely, hence $\cVal(h_{N_{T}},M) = \cVal(h_{k}, M)$.
Furthermore, by Lemma~\ref{lem:compatibility} we have that $h_k$ is compatible with $T_{n_k}$, where $n_k = {N_{T}} + k |\hat h^{loop}|_g$.
Thus, we have $T_{n_k}(h_k) = T_{N_{T}}(h_{N_{T}})$, and it follows that $T_{n_k}(h_k) \neq M(h_k)$.
Finally, since $h_k\rho$ is compatible with $T_{n_k}$, $\aVal(h_k, T_{n_k}) \leq p_i(h_k\rho) = p_i(h_{N_{T}}\rho) = \aVal(h_{N_{T}}, T_{N_{T}})$.
As $\cVal(h_k,M) = \cVal(h_{N_{T}},M) > \aVal(h_{N_{T}},T_{N_{T}})$, we also have $\cVal(h_k,M) > \aVal(h_k, T_{n_k})$.
That is, by Lemma~\ref{lem:dominance_histories}, $h_{k}$ is a non-dominance witness of $M$ by $T_{n_k}$.
There are infinitely many such $T_{n_k}$, thus we can conclude that $M \not \sqsubseteq (T_n)_{n\in \N}$.
\item
Finally, assume that $|\tilde{\rho}|_g > |G| |\mathcal{T}|$.
This means that there exists a state $(v, m^\mathcal{T})$ that appears at least twice in $\rho$ after following a green transition, hence before any red transition appears in $\rho$.
This yields a decomposition $\rho = \rho^1 \rho^{loop} \rho^2$ of $\rho$ such that $|\tilde \rho^{loop}|_g \geq 1$ and $\last { \tilde \rho^{loop}} = \last { \tilde \rho_1} = (v, m^\mathcal{T})$.
Let $k \in \N$ and consider the outcome $\rho_k = \rho^1 (\rho^{loop})^k \rho^2$.
by Lemma~\ref{lem:compatibility} we have that $h_{N_{T}} \rho_k$ is compatible with $T_{n_k}$, where $n_k = {N_{T}} + k |\tilde{h}^{loop}|_g$.
Thus, we have $T_{n_k}(h_{N_{T}}) = T_{N_{T}}(h_{N_{T}})$, and it follows that $T_{n_k}(h_{N_{T}}) \neq M(h_k)$.
Moreover, since $h_{N_{T}}\rho_k$ is compatible with $T_{n_k}$, $\aVal(h_{N_{T}}, T_{n_k}) \leq p_i(h_{N_{T}}\rho_k) = p_i(h_{N_{T}}\rho) = \aVal(h_{N_{T}}, T_{N_{T}})$.
As $\cVal(h_{N_{T}},M) > \aVal(h_{N_{T}},T_{N_{T}})$, we also have $\cVal(h_{N_{T}},M) > \aVal(h_{N_{T}}, T_{n_k})$.
That is, by Lemma~\ref{lem:dominance_histories}, $h_{{N_{T}}}$ is a non-dominance witness of $M$ by $T_{n_k}$.
Once again, there are infinitely many such $T_{n_k}$, thus we can conclude that $M \not \sqsubseteq (T_n)_{n\in \N}$.
\end{enumerate}
\end{itemize}
\end{itemize}
\end{proof}

\begin{proof}[Proof of Proposition~\ref{prop:bound}.\ref{prop:bound_comp2}]
Let us suppose that $(S_n)_{n \in \N} \not \sqsubseteq (T_n)_{n \in \N}$.
For every $n \in \N$, let $M_n := |G||\mathcal{T}|(|\mathcal{S}_n| + 1) + 1$,
where $\mathcal{S}_n$ denote the Mealy automaton derived from $\mathcal{S}$ that realizes the strategy $S_n$.
Then there exists $n \in \N$ satisfying $S_n \not \preceq (T_n)_{n \in \N}$,
hence, by Proposition~\ref{prop:bound}.\ref{prop:bound_comp1}, $S_n \not \preceq T_{M_n}$.
Let $k$ be the smallest integer satisfying $S_k \not \preceq T_{M_k}$.

If $k \leq N_{\mathcal{S}}$, we can conclude the proof immediatly.
Indeed, since $(S_n)_{n\in \N}$ is a chain by supposition, $S_{k} \preceq S_{N_{\mathcal{S}}}$,
and then $S_{k} \not \preceq (T_n)_{n \in \N}$ implies $S_{N_{\mathcal{S}}} \not \preceq (T_n)_{n \in \N}$.
Now, let us suppose, towards building a contradiction, that $k > N_{\mathcal{S}}$.
Note that, by %Lemma~\ref{lem:sinlge_strat_chain_dominance}
Proposition~\ref{prop:bound}.\ref{prop:bound_comp1}, we know that $S_k \not \sqsubseteq (T_n)_{n \in \N}$.
Therefore, for all $n \in \N$, we have $S_k \not \preceq T_n$.
In particular, $S_k \not \preceq T_{2M_k}$.
By Lemma~\ref{lem:dominance_histories}, there exists a non-dominance witness $h$ of $S_k$ by $T_{2M_k}$
such that $h$ is compatible with $S_k$ and $T_{2M_k}$, $S_k(h) \neq T_{2M_k}(h)$ and $\cVal(h, S_k) > \aVal(h, T_{2M_k})$.
Let $c= \cVal(h,S_k)$ and $a = \aVal(h,T_{2M_k})$.
Consider now a continuation path $\rho$ such that the outcome $h\rho$ is compatible with $S_k$ and $p_i(h\rho)=c$,
and a continuation path $\rho'$ such that the outcome $h\rho'$ is compatible with $T_{2M_k}$ and $p_i(h\rho')=a$.

We consider the valid path $\tilde{h}$ in $G \times \mathcal{T}$ that follows the behaviour of $S_k$ over the history $h$,
the valid path $\hat{h}$ in $\mathcal{S} \times G \times \mathcal{T}$ that follows the behaviour of $S_k$ and $T_{2M_k}$ over the history $h$,
and the valid path $\tilde{\rho}$ in $G \times \mathcal{S}$ that follows the behaviour of $S_k$ over the outcome $\rho$.
We distinguish the following two cases:
either $(a)$ $|\tilde h\tilde{\rho}|_g < N_{\mathcal{S}}$  or $(b)$ $|\tilde h\tilde{\rho}|_g \geq N_{\mathcal{S}}$.
\begin{itemize}
\item[$(a)$]
Assume that $|\tilde h\tilde{\rho}|_g < N_{\mathcal{S}}$.
Let $k' := k - |\tilde{h}\tilde{\rho}|_g$.
By Lemma~\ref{lem:compatibility}, we know that the outcome $h\rho$ is compatible with $S_{k'}$.
Thus, $\cVal(h,S_{k'}) \geq c$, and since $\cVal(h,S_{k'}) > c$ we have $\cVal(h,S_{k'}) > \aVal(h,T_{2M_k})$.
Furthermore, as $h\rho$ is compatible with $S_k$ and $S_{k'}$, we know that $S_{k'}(h) = S_{k}(h) \neq T_{2M_k}(h)$.
Hence, by Lemma~\ref{lem:dominance_histories}, the history $h$ is a non-dominance witness for $S_{k'}$ and $T_{2M_k}$.
Since $k' < k$, we have that $M_{k'} < M_k < 2M_k$, thus $S_{k'} \not \preceq T_{M_{k'}}$,
which contradicts the fact that $k$ is the smallest index such that $S_k \not \preceq T_{M_k}$.
\item[$(b)$]
Assume now that $|\tilde{h}\tilde{\rho}|_g = k$.
Since $k > N_{\mathcal{S}} = |G| |\mathcal{S}|(2|\mathcal{T}| + 1)$, at least one of the two following properties is verified:
either $|\tilde{h}|_g > 2|\mathcal{S}| |G| |\mathcal{T}|$ or $|\tilde{\rho}|_g > |G| |\mathcal{S}|$.
\begin{enumerate}
\item
Suppose that $|\tilde{h}|_g > 2|\mathcal{S}| |G| |\mathcal{T}|$.
Then there exists a state $(m^\mathcal{S},v,m^\mathcal{T})$ that appears three times in $\hat{h}$ after the $\mathcal{S}$-component follows a green transition.
Thus we can decompose $h$ as follows:
$h = h^{1}h^{loop}_1h^{loop}_2h^{2}$, such that if we consider the corresponding decomposition
$\hat{h}^{1}\hat{h}_1^{loop}\hat{h}_2^{loop}\hat{h}^{2}$ of $\hat{h}$, we have
$\last{\hat{h}^{1}}=\last{\hat{h}_1^{loop}}=\last{\hat{h}_2^{loop}} = (m^\mathcal{T},v,m^\mathcal{S})$.
Furthermore, there exists at least one green edge in the $\mathcal{S}$-component of each path $\hat{h}_1^{loop}$ and $\hat{h}_2^{loop}$.
Let $k_1 = |\hat{h}_1^{loop}|^\mathcal{S}_g$ and $k_2 = |\hat{h}_2^{loop}|^\mathcal{S}_g$.

By Lemma~\ref{lem:compatibility}, since $\hat{h}$ is compatible with $T_{2M_k}$,
we have $|\hat{h}|^\mathcal{T}_g \leq 2M_k$, hence either $|\hat{h}_1^{loop}|^\mathcal{T}_g \leq M_k$ or $|\hat{h}_2^{loop}|^\mathcal{T}_g \leq M_k$.
Assume $|\hat{h}_1^{loop}|^\mathcal{T}_g \leq M_k$ (respectively $|\hat{h}_2^{loop}|^\mathcal{T}_g \leq M_k$).
Consider the path $h_1h_2^{loop}h_2$
(respectively $h_1h_1^{loop}h_2$).
It is a valid path since $\last{\hat{h}_1}= \last{\hat{h}_1^{loop}}=(m^\mathcal{S},v, m^\mathcal{T})$
(respectively $\last{\hat{h}_1^{loop}}=\last{\hat{h}_2^{loop}}= (m^\mathcal{S},v, m^\mathcal{T})$).
Since $|\hat{h}_1^{loop}|^\mathcal{T}_g \leq M_k$
(respectively $|\hat{h}_2^{loop}|^\mathcal{T}_g \leq M_k$),
we have that $| \hat{h}_1 \hat{h}_2^{loop}  \hat{h}_2|^\mathcal{T}_g \geq |\hat{h}|^\mathcal{T}_g - M_k$
(respectively $| \hat{h}_1 \hat{h}_1^{loop}  \hat{h}_2|^\mathcal{T}_g \geq |\hat{h}|^\mathcal{T}_g - M_k$).
Thus by Lemma~\ref{lem:compatibility}, $h_1h_2^{loop}h_2$ (respectively $h_1h_1^{loop}h_2$) is compatible with $T_{M_k}$.
Furthermore, the outcome $h_1h_2^{loop}h_2\rho'$ (respectively $h_1h_1^{loop}h_2\rho'$) is compatible with $T_{M_k}$ and yields payoff $a$.
Hence, we have that $\aVal(h_1h_2^{loop}h_2, T_{M_k}) \leq a$ (respectively $\aVal(h_1h_1^{loop}h_2 ,T_{M_k}) \leq a$).
Let $k' = k- k_1$ (respectively $k' = k-k_2$).
By Lemma~\ref{lem:compatibility}, we know that $h_1h_2^{loop}h_2$
(respectively $h_1h_1^{loop}h_2$) is compatible with $S_{k'}$.
Furthermore, the outcome $h_1h_2^{loop}h_2\rho$
(respectively $h_1h_1^{loop}h_2\rho$) is compatible with $S_{k'}$ and yields payoff $c$.
Hence, we have that $\cVal(h_1h_2^{loop}h_2, S_{k'}) \geq c$
(respectively $\cVal(h_1h_1^{loop}h_2 ,S_{k'}) \geq c$).
Finally, we know that $S_{k'}(h_1h_2^{loop}h_2) \neq T_{M_k}(h_1h_2^{loop}h_2)$
(respectively $S_{k'}(h_1h_1^{loop}h_2) \neq T_{M_k}(h_1h_1^{loop}h_2)$)
since $S_{k'}(h_1h_2^{loop}h_2)=S_{k'}(h)$
(respectively $S_{k'}(h_1h_1^{loop}h_2)=S_{k'}(h)$)
and $T_{M_k}(h_1h_2^{loop}h_2)=T_{2M_k}(h)$
(respectively $T_{M_k}(h_1h_1^{loop}h_2)=T_{2M_k}(h)$).
By Lemma~\ref{lem:dominance_histories}, the history $h_1h_2^{loop}h_2$
(respectively $h_1h_1^{loop}h_2$)
is a non-dominance witness for $S_{k'}$ and $T_{M_k}$.
Since $k' < k$, we have that $M_{k'} < M_k$, thus $S_{k'} \not \preceq T_{M_{k'}}$, which is a contradiction with the fact that $k$ is the smallest index such that $S_k \not \preceq T_{M_k}$.

\item
Suppose that $|\tilde{\rho}|_g > |G| |\mathcal{S}|$.
Then there exists a state $(v,m^\mathcal{S})$ that appears twice in $\tilde{\rho}$ after following a green transition.
Thus we can decompose $\rho$ as follows:
$\rho = \rho^{1}\rho^{loop}\rho^{2}$, such that if we consider the corresponding decomposition
$\tilde{\rho}^{1}\tilde{\rho}^{loop}\tilde{\rho}^{2}$ of $\tilde{\rho}$, we have
$\last{\tilde{\rho}^{1}}=\last{\tilde{\rho}^{loop}} = (v,m^\mathcal{S})$ and $|\tilde{\rho}^{loop}|_g > 0$.
Then $h \rho^{1} \rho^{2}$ is a valid outcome since $\last{\tilde{\rho}^{1}}=\last{\tilde{\rho}^{loop}}$.
Furthermore, $p_i(h \rho^{1} \rho^{2})=c$.
Let $k' = |  \tilde{h} \tilde{\rho^{1}} \tilde{\rho^{2}} |_g$.
As $|\tilde{\rho}^{loop}|_g > 0$, we have $k' < k$.
By Lemma~\ref{lem:compatibility}, we know that the outcome $h \rho^{1} \rho^{2}$ is compatible with $S_{k'}$.
Thus, $\cVal(h, S_{k'}) \geq c$.
Finally, since $h$ is a prefix of $h \rho^{1} \rho^{2}$, we know that $S_{k'}(h) = S_k(h)$ hence $S_{k'}(h) \neq T_{2M_k}(h)$.
Hence, by Lemma~\ref{lem:dominance_histories}, the history $h$ is a non-dominance witness for $S_{k'}$ and $T_{2M_k}$.
Since $k' < k$, we have that $M_{k'} < M_k < 2M_k$, thus $S_{k'} \not \preceq T_{M_{k'}}$, which is a contradiction with the fact that $k$ is the smallest index such that $S_k \not \preceq T_{M_k}$.
\end{enumerate}
\end{itemize}
\end{proof} 

Since the property $\textsf{P}_1$ can be decided in $\textsc{PTime}$ by applying Lemma~\ref{lem:comparison_two_strategies}
with adequately chosen Mealy automata as parameters, we obtain the following theorem.

\begin{theorem}\label{theo:assumption_1}
Given a generalised safety/reachability game and a parameterized automaton, we can decide in $\textsc{PTime}$ whether the automaton realizes a chain of strategies.
\end{theorem}

Similarly, the property $\textsf{P}_3$ can be decided in $\textsc{PTime}$ by applying Lemma~\ref{lem:comparison_two_strategies}
with $\mathcal{M}$ and the Mealy automaton corresponding to the strategy $T_{N_{T}}$ as parameters.
Moreover, by Proposition~\ref{prop:bound}.\ref{prop:bound_comp1}, the problem of deciding property $\textsf{P}_4$
can be reduced in polynomial time to the problem of deciding property $\textsf{P}_3$.
Therefore Proposition~\ref{prop:bound}.\ref{prop:bound_comp2} implies our final decidability result.

%\begin{lemma}\label{lem:sinlge_strat_chain_dominance}
%Given a generalised safety/reachability game, a Mealy automaton and a parameterized automaton realizing a uniform chain of strategies,
%we can decide whether the strategy realized by the first is dominated by the chain realized by the second.
%\end{lemma}

\begin{theorem}\label{theo:uniform_dominance}
Given a generalised safety/reachability game and two parameterized automata realizing uniform chains of strategies,
we can decide in $\textsc{PTime}$ whether the chain realized by the first is dominated by the one from the second.
\end{theorem}

\section{Conclusions and outlook}
We have observed that admissibility is lacking as a rationality criterion for infinite sequential games with quantitative payoffs. Our primary counterexample suggests that chains of strategies could provide a suitable framework to circumvent this issue. Abstract order-theoretic considerations revealed that in the most general case, this does not work. However, if we restrict to countable collections of strategies, every chain is below a maximal chain. This restriction is very natural in a TCS setting. A more in-depth exploration of the game-theoretic merits of such a notion of rationality based on chains of strategies is left for the future.

We explored the abstract approach in the concrete setting of generalized safety/reachability games. Here, parameterized automata can give a very concrete meaning to chains of strategies. Several fundamental algorithmic questions are decidable in $\textsc{PTime}$. There are more algorithmic questions to investigate. Moreover, the generalization of our results from generalized safety/reachability games to games with $\omega$-regular objectives seems achievable - our proofs make only very limited use of the special features of the former. Both these endeavours could benefit from a better understanding of parameterized automata in general.

\bibliography{biblio}

\end{document}